\documentclass[11pt]{article}

\usepackage{times}
\usepackage{amsmath}
\usepackage{amssymb}
\usepackage{latexsym}
\RequirePackage{pstricks}
\RequirePackage{pst-node}
\RequirePackage{pst-plot}
\RequirePackage{pst-tree}
\usepackage{graphicx}
\usepackage{hyperref}
\usepackage{picins,boxedminipage}

\usepackage[dvips]{epsfig}
\usepackage{latexsym}
\usepackage{color}

\advance \topmargin by -0.7in
\newtheorem{theorem}{Theorem}[section]
\newtheorem{proposition}[theorem]{Proposition}
\newtheorem{lemma}[theorem]{Lemma}
\newtheorem{corollary}[theorem]{Corollary}

\newtheorem{fact}[theorem]{Fact}
\newtheorem{claim}[theorem]{Claim}
\newtheorem{con}[theorem]{Conjecture}
\newcommand{\sur}{\sigma }
\newcommand{\qed}{\mbox{}\hspace*{\fill}\nolinebreak\mbox{$\rule{0.6em}{0.6em}
$}}
\newenvironment{proof}{\noindent {\bf Proof:}}{$\qed$\par\medskip}

\newcommand{\cR}{{\cal R}}
\newcommand{\cZ}{{\cal Z}}

\begin{document}

\title{
  Flow-Cut Gaps for Integer and Fractional Multiflows
%  Multiflows in series-parallel Graphs and an Extension of the Okamura
%  Seymour Theorems
}
\author{
Chandra Chekuri\thanks{Dept.\ of Computer Science, 201 N.\ Goodwin Ave.,
Univ.\ of Illinois,
Urbana, IL 61801, USA. Email: {\tt chekuri@cs.illinois.edu}}
\and
F.~Bruce Shepherd\thanks{Department of Mathematics and Statistics, McGill University,
805 Sherbrooke West, Montreal, QC, Canada. Email: {\tt
bruce.shepherd@mcgill.ca}}
\and
Christophe Weibel\thanks{Department of Mathematics and Statistics, McGill University,
805 Sherbrooke West, Montreal, QC, Canada. Email:
{\tt christophe.weibel@gmail.com}}
}

\begin{titlepage}
\maketitle
\def\thepage {} % Kill pagenumbering
\thispagestyle{empty}

\begin{abstract}
  Consider a {\em routing problem} instance consisting of a {\em
    demand graph} $H=(V,E(H))$ and a {\em supply graph} $G=(V,E(G))$.
  If the pair obeys the cut condition, then the {\em flow-cut gap} for
  this instance is the minimum value $C$ such that there exists a
  feasible multiflow for $H$ if each edge of $G$ is given capacity
  $C$. It is well-known that the flow-cut gap may be greater than $1$
  even in the case where $G$ is the (series-parallel) graph $K_{2,3}$.
  In this paper we are primarily interested in the ``integer''
  flow-cut gap. What is the minimum value $C$ such that there exists a
  feasible integer valued multiflow for $H$ if each edge of $G$ is
  given capacity $C$?  We formulate a conjecture that states that the
  integer flow-cut gap is quantitatively related to the fractional
  flow-cut gap. In particular this strengthens the well-known
  conjecture that the flow-cut gap in planar and minor-free graphs is
  $O(1)$ \cite{GNRS99} to suggest that the integer flow-cut gap is
  $O(1)$.  We give several technical tools and results on non-trivial
  special classes of graphs to give evidence for the conjecture and
  further explore the ``primal'' method for understanding flow-cut
  gaps; this is in contrast to and orthogonal to the highly successful
  metric embeddings approach. Our results include the following:
  \begin{itemize}
  \item Let $G$ be obtained by series-parallel operations starting
    from an edge $st$, and consider orienting all edges in $G$ in the
    direction from $s$ to $t$. A demand is {\em compliant} if its
    endpoints are joined by a directed path in the resulting oriented
    graph.  We show that if the cut condition holds for a
    compliant instance and $G+H$ is Eulerian, then an integral
    routing of $H$ exists. This result includes as a special case,
    routing on a ring but is not a special case of the Okamura-Seymour
    theorem.
  \item Using the above result, we show that the integer flow-cut gap
    in series-parallel graphs is $5$.  We also give an explicit class
    of routing instances that shows via elementary calculations that
    the flow-cut gap in series-parallel graphs is at least $2-o(1)$; this
    is motivated by and simplifies the proof by Lee and Raghavendra
    \cite{LeeR07}.
  \item The integer flow-cut gap in $k$-Outerplanar graphs is $c^{O(k)}$ for
   some fixed constant $c$.
 \item A simple proof that the flow-cut gap is $O(\log k^*)$
   where $k^*$ is the size of a node-cover in $H$; this was
   previously shown by G\"{u}nl\"{u}k via a more intricate proof
   \cite{Gunluk07}.
  \end{itemize}
\end{abstract}

\end{titlepage}
\pagenumbering{arabic}

\section{Introduction}
Given a (undirected) graph $G=(V,E)$ a {\em routing} or {\em
  multiflow} consists of an assignment $f:{\cal P} \rightarrow R_+$
where ${\cal P}$ is the set of simple paths in $G$ and such that for
each edge $e$, $\sum_{P \in P(e)} f_P \leq 1$, where $P(e)$ denotes
the set of paths containing $e$. Given a {\em demand graph}
$H=(V,E(H))$ such a routing {\em satisfies} $H$ if $\sum_{P \in
  P(u,v)} f_P = 1$ for each $g=uv \in E(H)$, where $P(u,v)$ denotes
paths with endpoints $u$ and $v$ (one may assume a simple demand graph
without loss of generality). If such a flow exists, we call the
instance {\em routable}, or say $H$ is routable in $G$. Edges of $G$
and $H$ are called {\em supply edges} and {\em demand edges}
respectively.  The above notions extend naturally if each supply edge
$e$ is equipped with a capacity $u_e$ and each demand edge $g$ is
equipped with a demand $d_g$. If $u$ is an integral vector, we denote by $G_u$, the graph
obtained by making $u_e$ copies of each edge $e$. $H_d$ is defined similarly.
We call the routing $f$ {\em integral}
(resp. {\em half-integral}) if each $f_P$ (resp. $2 f_P$) is an
integer.

\parpic[r]{
  \begin{minipage}{0.22\linewidth}
    \begin{center}
      \includegraphics[height=1in]{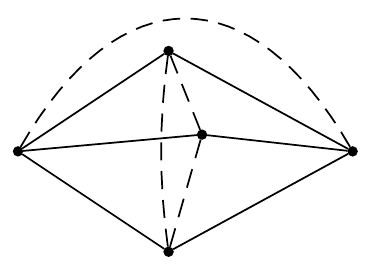}
    \end{center} \label{fig:k23}
%\centerline{\input{comp.pstex_t}}
%\caption{A series-parallel instance satisfying the cut condition,
%but with no multiflow (fractional  routing).}
  \end{minipage}
}
For any set $S \subseteq V$ we denote by $\delta_G(S)$ the set of
edges with exactly one end in $S$, and the other in $V-S$. We define
$\delta_H(S)$ similarly. (For graph theory notation we primarily
follow Bondy and Murty \cite{BM}.)  The supply graph $G$ satisfies the
{\em cut condition} for the demand graph $H$ if $|\delta_G(S)| \ge
|\delta_H(S)|$ for each $S \subset V$. We sometimes say that the pair
$G,H$ satisfies the cut-condition.  Clearly the cut condition is a
necessary condition for the routability of $H$ in $G$.  The
cut-condition is not sufficient as shown by the well-known example
where $G = K_{2,3}$ is a series-parallel graph with a demand graph
(in dotted edges) as shown in the figure.

Given a graph $G$ and a real number $\alpha > 0$ we use $\alpha G$ to
refer to the graph obtained from $G$ by multiplying the capacity of
each edge of $G$ by $\alpha$. Given an instance $G,H$ that satisfies
the cut-condition, the {\em flow-cut gap} is defined as the smallest
$\alpha \ge 1$ such that $H$ is routable in $\alpha G$; we also refer
to $\alpha$ as the {\em congestion}. We denote this quantity by
$\alpha(G,H)$.  Traditional combinatorial optimization literature has
focused on characterizing conditions under which the cut-condition is
sufficient for (fractional, integral or half-integral) routing, in
other words the setting in which $\alpha(G,H) = 1$; see
\cite{Schrijver_book} for a comprehensive survey of known results.
Typically, these characterizations involve both the supply and demand
graphs. A prototypical result is the Okamura-Seymour Theorem
\cite{OS} that states that the cut-condition is sufficient for a
half-integral routing if $G$ is a planar graph and all edges of $H$
are between the nodes of a single face of $G$ in some planar
embedding.  The proofs of such result rely on what we will term
``primal-methods'' in that they try to directly exhibit routings of
the demands, rather than appealing to dual solutions.  \iffalse A
standard proof technique underlying these results is {\em pushing}
wherein a demand edge $uv$ is replaced by two demand edges $uw$ and
$wv$ while preserving the cut-condition. This allows an inductive
reduction process and leads to integral flows.  We loosely refer to
this type of proof technique the ``primal''-method for proving
flow-cut results.  \fi

On the other hand, since the seminal work of Leighton and Rao
\cite{LeightonR88} on flow-cut gaps for uniform and product multiflow
instances, there has been an intense focus in the algorithms and
theoretical computer science community on understanding flow-cut gap
results for classes of graphs. This was originally motivated by the
problem of finding (approximate) sparse cuts. A fundamental and
important connection was established in \cite{LLR,AR} between flow-cut
gaps and metric-embeddings. More specifically, for a graph $G$, let
$\alpha(G)$ be the largest flow-cut gap over all possible capacities
on the edges of $G$ and all possible demand graphs $H$. Also let
$c_1(G)$ denote the maximum, over all possible edge lengths on $G$, of
the minimum {\em distortion} required to embed the finite metric on
the nodes of $G$ (induced by the edge lengths) into an $\ell_1$-space.
Then the results in \cite{LLR,AR} showed that $\alpha(G) \le c_1(G)$
and subsequently \cite{GNRS99} showed that $\alpha(G) = c_1(G)$.
Using Bourgain's result that $c_1(G) = O(\log |V|)$ for all $G$,
\cite{LLR,AR} showed that $\alpha(G) = O(\log |V(G)|)$, and further
refined it to prove that $\alpha(G,H) = O(\log |E_H|)$. Numerous
subsequent results have explored this connection to obtain a variety
of flow-cut gap results. The proofs via metric-embeddings are
``dual''-methods since they work by embedding the metric induced by
the dual of the linear program for the maximum concurrent
multicommodity flow. The embedding approach has been successful in
obtaining flow-cut gap results (amongst several other algorithmic
applications) as well as forging deep connections between various
areas of discrete and continuous mathematics. However, this approach
does not directly give us integral routings even in situations when
they do exist.

In this paper we are interested in the {\em integer} flow-cut gap in
undirected graphs. Given $G,H$ that satisfy the cut-condition, what
is the smallest $\alpha$ such that $H$ can be integrally routed in
$\alpha G$? Is there a relationship between the (fractional)
flow-cut gap and the integer flow-cut gap? A result of Nagamochi and
Ibaraki relates the two gaps in {\em directed} graphs. Let $G=(V,A)$
and $H=(V,R)$ be a supply and demand digraph, respectively. We call
$(G,H)$ {\em cut-sufficient} if for each capacity function $u: A
\rightarrow \cZ^+$ and demand function $d:R \rightarrow \cZ^+$,
$G_u,H_d$ satisfying the cut-condition implies the existence of a
fractional multiflow for $H_d$ in $G_u$.
%We sometimes refer to $G_u$ as the multigraph
%obtained from $G$ by taking $u(e)$ parallel copies of $e$. Similarly
%for $H_d$.

\begin{theorem}[\cite{NagamochiI89}]
\label{thm:nagiba}
If $(G,H)$ is cut-sufficient, then for any integer capacity vector $u$
and integer demand vector $d$ such that $G_u,H_d$ satisfy the cut
condition, there is an integer multiflow for $H_d$ in $G_u$.

\end{theorem}

The above theorem does not extend to the undirected case. Consider
taking $G$ to be a cycle and $H$ to be a complete graph. Then it is
known that $(G,H)$ is cut-sufficient but we are not guaranteed an
integral flow for integer valued $u$ and $d$; for example when $G$
is a $4$ cycle with unit capacities and $H$ consists of two crossing
edges with unit demands. For integer valued $u$ and $d$, however,
there is always a half-integral routing of $H_d$ in $G_u$ whenever
$(G_u,H_d)$ satisfies the cut-condition. We may therefore ask if a
weaker form of Theorem~\ref{thm:nagiba} holds in undirected graphs. Namely,
where we only ask for half-integral flow instead of integral flows.
%Although we know of no counterexamples, the proof technique of
%Theorem~\ref{thm:nagiba} does not seem to directly apply.
One case where one does get such a half-integral routing in
undirected graphs is the following. Consider the case when $G=H$; if
the pair $(G,G)$ is cut-sufficient we simply say that $G$ is
cut-sufficient. It turns out that this is precisely the class of
$K_5$-minor free graphs (Seymour \cite{seymour1981a}; cf. Corollary
75.4d \cite{Schrijver_book}). Moreover we have the following.

\iffalse
The above theorem does not extend to the undirected case. Consider
taking $G$ to be a cycle of length $4$ and $H$ to consist of two
crossing edges. The cut condition for any weighting of $G,H$ is
sufficient for a flow to exist. However, with weights $1$, there
does not exist an integral routing of $H$ in $G$. In this case,
however there is always a half-integral routing of $H$ in $G$ if the
cut condition holds. It turns out that this is true in general. That
is, we may naturally extend the notion of cut-sufficiency to
undirected pairs $(G,H)$.
%Using results of Geelen and Guenin,
%Schrijver characterizes the pairs which are cut-sufficient. $(G,H)$
%is cut-sufficient if and only if $G+H$ has no odd $K_5$ minor (all
%triangles of the $K_5$ are odd) in the signed graph with $E(G)$
%being even edges, and $E(H)$ being odd edges.
If the pair $(G,G)$ is cut-sufficient, then we call $G$ itself
cut-sufficient. It turns out that this is precisely the class of
$K_5$-minor free graphs (Seymour; cf. Corollary 75.4d
\cite{Schrijver_book}). Moreover we have the following.
\fi

\begin{theorem}[Seymour]
  If $G$ is cut-sufficient, then for any nonnegative integer
  weightings $u,d$ on $E(G)$ for which $G_u,G_d$ satisfies the cut
  condition, there is a half-integral routing of $G_d$ in
  $G_u$. Moreover, if $G_u+G_d$ is Eulerian, then there is an integral
  routing of $G_d$.
\end{theorem}

In this paper we ask more broadly, whether the fractional and
integral flow-cut gaps are related even in  settings where the
flow-cut gap is greater than $1$. We formulate the conjecture below.

\begin{con}[Gap-Conjecture]
  \label{conj:gap} Does there exist a global constant $C$ that
  satisfies the following.  Let $G=(V,E)$ and $H=(V,R)$ be a supply
  and demand graph respectively.  Suppose that for each capacity
  function $u: A \rightarrow \cZ^+$ and demand requirement $d:R
  \rightarrow \cZ^+$, $G_u,H_d$ satisfy the cut-condition implies that
  there is a fractional multiflow for $H_d$ in $G_d$ with congestion
  $\beta$. Then this implies that there in an integer multiflow for
  $H_d$ in $G_d$ with congestion $C \beta$.
\end{con}

\noindent
We do not currently know if the statement holds for $C=2$ in instances
with $\beta=1$ (thus generalizing Seymour's theorem mentioned above).
There are several natural weakenings of the conjecture that are
already unknown. For instance, one may allow $C$ to depend on a class
of instances (such as planar or series parallel supply graphs). More
generally, it would be of interest to bound the integer flow cut gap
as some $g(\beta)$, e.g., $g(\beta)=O(\text{poly}(\beta))$.
Previously other conjectures relating fractional and integer
multiflows were shown to be false. For instance, Seymour conjectured
that if there is a fractional multiflow for $G,H$, then it implies a
half-integer multiflow. These conjectures have been strongly disproved
(see \cite{Schrijver_book}). Note that our conjecture differs from the
previous ones in that we relate the flow-cut gap values for hereditary
classes of instances on $G,H$.

The Gap-Conjecture has several important implications. First, it
would give structural insights into flows and cuts in graphs.
Second, it would allow fractional flow-cut gap results obtained via
the embedding-based approaches to be translated into integer
flow-cut gap results. Finally, it would also shed light on the
approximability of the congestion minimization problem in special
classes of graphs. In congestion minimization we are given $G,H$ and
are interested in the least $\alpha$ such that $\alpha G$ has an
integer routing for $H$. Clearly, the congestion required for a
fractional routing is a lower bound on $\alpha$; moreover this lower
bound can be computed in polynomial time via linear programming.
Almost all the known approximation guarantees are with respect to
this lower bound; even in directed graphs an $O(\log n/\log \log n)$
approximation is known via randomized rounding \cite{RaghavanT87}.
In general undirected graphs, this problem is hard to approximate to
within an $\Omega(\log \log n)$-factor \cite{AndrewsZ05}.  However,
for planar graphs and graphs that exclude a fixed minor, it is
speculated that the problem may admit an $O(1)$ approximation. The
Gap-Conjecture relates this to the conjecture of Gupta et al.\
\cite{GNRS99} that states that the fractional flow-cut gap is $O(1)$
for all graphs that exclude a fixed minor. Thus the congestion
minimization problem has an $O(1)$ approximation in planar graphs if
the Gupta et al.\ conjecture and the Gap-Conjecture are both true.
We also note that an $O(1)$ gap between fractional and integer
multiflows in planar graphs (or other families of graphs) would shed
light on the Gap-Conjecture.

Our current techniques seem inadequate to resolve the
Gap-Conjecture. It is therefore natural to prove the Gap-Conjecture in
those settings where we do have interesting and non-trivial upper
bounds on the (fractional) flow-cut gap. Note that the conjecture
follows easily when $G$ and $H$ are unrestricted (complete graphs); in
this case the flow-cut gap is $\Omega(\log n)$; one may consider $G$,
a bounded degree expander, with $H$, a uniform multiflow
\cite{LeightonR88}. On the other side, randomized rounding shows that
the integer flow-cut gap is $O(\log n)$. Now, if $G$ is a complete
graph and $H$ is a complete graph on a subset of $k$ nodes of $G$ then
the flow-cut gap for such instances is $\Omega(\log k)$; this easily
follows from the above mentioned general result. For these instances
the fractional flow-cut gap improves to $O(\log k)$, as shown in
\cite{LLR,AR}. One can obtain an improvement for the integer flow-cut
gap but one cannot employ simple randomized rounding; in \cite{CKS-unpub}
it is shown that the integer flow-cut gap for these instances
is $O(\text{polylog}(k))$; this relies on the results in
\cite{well-linked,KhandekarRV}.
%It was observed in
%\cite{CKS-unpub} that the integer flow-cut gap is
%$O(\text{polylog}(|E_H|))$; this relies on the results in
%\cite{well-linked,KhandekarRV}. Recall that \cite{LLR,AR} show that
%the flow-cut gap is $\Theta(\log |E_H|)$.

In a sense, the Gap-Conjecture is perhaps more relevant and interesting
in those cases where the flow-cut gap is $O(1)$. We focus on
series-parallel graphs and $k$-Outerplanar graphs for which we know
flow-cut gaps of $2$ \cite{ChakrabartiJLV08} and $c^k$ (for some
universal constant $c$) \cite{ChekuriGNRS} respectively. Proving
flow-cut gaps for even these restricted families of graphs has taken
substantial technical effort. In this paper we affirm that one can
prove similar bounds for these graphs for the integer flow-cut
gap. For instance, in series parallel instances, we show that the
integer flow-cut gap is at most $5$ (and we conjecture it is $2$).

\medskip
\noindent {\bf Overview of results and techniques:}
In this paper we focus especially on applying primal methods to two classes of graphs for which
the flow-cut gap is known to be $O(1)$: series parallel graphs and
$k$-Outerplanar graphs.

The first proof that series parallel instances had a constant flow-cut
was given in \cite{GNRS99}; subsequently a gap of $2$ was shown in
\cite{ChakrabartiJLV08}. This latter upper bound is tight since it is
shown in \cite{LeeR07} that there are instances where the gap is
arbitrarily close to $2$.  We give a simpler proof of the lower bound
in this paper that is based on an explicit (recursive) instance and
elementary calculations --- our proof is inspired by \cite{LeeR07} but
avoids their advanced metric-embeddings machinery.

In Section~\ref{sec:spcong5} we show that for series-parallel graphs
the integer flow-cut gap is at most $5$. The primal-method has
generally been successful in identifying restrictions on demand
graphs for which the cut-condition implies routability. We follow
that approach and identify several classes of demands in
series-parallel graphs  for which cut-condition implies routability
(see Sections~\ref{sec:compliant} - ~\ref{sec:semicompliant}). The
main class exploited to obtain the congestion 5 result,  are the
so-called compliant demands (Section~\ref{sec:semicompliant}).
However, the critical base case for compliant demands boils down to
determining classes of demands on $K_{2m}$ instances for which the
cut condition implies routability. In fact, for $K_{2m}$ instances,
we are able to give a complete characterization of demand graphs $H$
for which $(K_{2m},H)$ is cut sufficient - see
Section~\ref{subsec:odd}. Moreover, we conjecture that this
forbidden minor characterization carries over for all
series-parallel graphs (it does not completely characterize
cut-sufficiency in general however).

 One
ingredient we use is a general proof technique for ``pushing''
demands similar to what has been used in previous primal proofs; for
instance in the proof of the Okamura-Seymour theorem \cite{OS}. We
try to replace a demand edge $uv$ by a pair of edges $ux,xv$ to make
the instance simpler (we call this {\em
  pushing to} $x$). Failing to push, identifies some tight cuts and
sometimes these tight cuts can be used to shrink to obtain an instance
for which we know a routing exists.  This contradiction means that we
could have pushed in the first place.

In \cite{ChekuriGNRS}, an upper bound of $c^k$ (for some constant
$c$) is given for the flow-cut gap in $k$-Outerplanar graphs. In
this paper (Section~\ref{sec:kouter}), we show that the integer
flow cut gap in this case is $c^{O(k)}$. In this effort, we
explicitly employ a second proof ingredient which is a simple {\em
rerouting} lemma that was stated and used in \cite{treewidth} (see
Section~\ref{sec:rerouting}). Informally speaking the lemma says the
following. Suppose $H$ is a demand graph and for simplicity assume
it consists of pairs $s_1t_1, \ldots, s_kt_k$. Suppose we are able
to route the demand graph $H'$ consisting of the edges $s_1s_1',
t_1t_1', \ldots, s_ks_k',t_kt_k'$ in $G$ where $s_1',t_1',\ldots,
s_k',t_k'$ are some arbitrary intermediate nodes. Let $H''$ be the
demand graph consisting of $s'_1t'_1,\ldots, s_k't_k'$. The lemma
states that if $G,H$ satisfies the cut-condition and the
aforementioned routing exists in $G$ then $2G,H''$ satisfies the
cut-condition. Clearly we can compose the routings for $H'$ and
$H''$ to route $H$. The advantage of the lemma is that it allows us
to reduce the routing problem on $H$ to that in $H''$ by choosing
$H'$ appropriately. This simple lemma and its variants give a way to
prove approximate flow-cut gaps effectively.

The rerouting lemma sometimes leads to very simple and insightful
proofs for certain results that may be difficult to prove via other
means --- see \cite{treewidth}.  In this paper we give two
applications of the lemma. We give (in Section~\ref{sec:gunluk}) a
very short and simple proof of a result of G\"{u}nl\"{u}k
\cite{Gunluk07}; he refined the result of \cite{LLR,AR} and showed
that $\alpha(G,H) = O(\log k^*)$ where $k^*$ is the node-cover size of
$H$. Clearly $k^* \le |E_H|$ and can be much smaller. We also show
that the integer flow-cut gap for $k$-Outerplanar graphs is $c^{O(k)}$
for some universal constant $c$; in fact we show a slightly stronger
result (see Section~\ref{sec:kouter}). Previously it was known that
the (fractional) flow-cut gap for $k$-Outerplanar graphs is $c^k$
\cite{ChekuriGNRS}.

% We also derive an improved result (Section~\ref{sec:kouter})
%on the flow-cut gap of planar instances in which the terminals lie on
%a fixed number of faces.

Our integer flow-cut gap results imply corresponding new approximation
algorithms for the congestion minimization problem on the graph
classes considered. Apart from this immediate benefit, we feel that it
is important to complement the embedding-based approaches to
simultaneously develop and understand corresponding tools and
techniques from the primal point of view. As an example, Khandekar,
Rao and Vazirani \cite{KhandekarRV}, and subsequently \cite{OSVV08},
gave a primal-proof of the Leighton-Rao result on product
multicommodity flows \cite{LeightonR88}. This new proof had
applications to fast algorithms for finding sparse cuts
\cite{KhandekarRV,OSVV08} as well as approximation algorithms for the maximum
edge-disjoint path problem \cite{RaoZ}.

\section{Basics and Notation}

We first discuss some basic and standard reduction
operations in primal proofs for flow-cut gaps and also set up the
necessary notation for series-parallel graphs.

\subsection{Some Basic Operations Preserving the Cut Condition}
\label{sec:basicops}

We present several operations that turn an instance
$G,H$ satisfying the cut condition into smaller instances with the
same property. We call an instance $G,H$ {\em Eulerian} if
$G+H$ is Eulerian; we also seek to preserve this property.

For $S \subseteq V$, the capacity of the cut $\delta_G(S)$, is just
$|\delta_G(S)|$ (or sum of capacities if edges have capacities).
Similarly, the demand of such a cut is $|\delta_H(S)|$. Hence the
surplus is $\sur(S) = |\delta_G(S)|- |\delta_H(S)|$. The set $S$, and
cut $\delta(S)$, is called {\em tight} if $\sur(S)=0$. The {\em cut
  condition} is then satisfied for an instance $G,H$ if $\sur(S) \geq
0$ for all sets $S$. One may naturally obtain ``smaller'' routing
instances from $G,H$ by performing a contraction of a subgraph of $G$
(not necessarily a connected subgraph) and removing loops from the
resulting $G'$, and in the resulting demand graph $H'$. It is easily
checked that if $G,H$ has the cut condition, then so does any
contracted instance.

We call a subset $A \subseteq V(G)$ {\em central} if both $G[A]$ and
$G[V-A]$ are connected. The following is well-known
cf. \cite{Schrijver_book}.

\begin{lemma}
\label{lem:minimal}
$G,H$ satisfy the cut condition if and only if the surplus of every central set
is nonnegative.
%Given any pair $G,H$, if the surplus of every central set is
%nonnegative, then the instance satisfies the cut condition.
\end{lemma}

\smallskip \noindent {\bf $1$-cut reduction:} This operation
takes an instance where $G$ has a cut node $v$ and consists of
splitting $G$ into nontrivial pieces determined by the components of
$G-v$. Demand edges $f$ with endpoints $x,y$ in distinct components
are replaced by two demands $xv,yv$ and given over to the obvious
instance.  One easily checks that each resulting instance again
satisfies the cut condition.  A simple argument also shows that the
Eulerian property is maintained in each instance if the original
instance was Eulerian.

\smallskip \noindent {\bf Parallel reduction:} This takes as
input an instance with a demand edge $f$ and supply edge $e$, with
the same endpoints.  The reduced instance is obtained by simply
removing $f,e$ from $H$ and $G$ respectively.  Trivially the new
instance satisfies the cut condition and is Eulerian if $G,H$ was.
%The parallel reduction is a special case of the path reduction
%defined and used later in Section~\ref{sec:cong3}.

\smallskip \noindent {\bf Slack reduction:} This works on an instance
where some edge $e$ (in $G$ or $H$) does not lie in any tight cut. In
this case, if $e \in G$, we may remove $e$ from $G$ and add it to $H$.
If $e \in H$, we may add two more copies of $e$ to $H$.  Again, this
trivially maintains the cut condition and the Eulerian property.

\smallskip \noindent {\bf Push operations:} Such an operation is
usually applied to a demand edge $xy$ whose endpoints lie in
distinct components of $G-\{u,v\}$ for some $2$-cut $u,v$. {\em
Pushing a demand $xy$ to $u$} involves replacing the demand edge
$xy$ by the two new demands $xu,uy$.  Such an operation clearly
maintains the Eulerian property but it may not maintain the cut
condition. We have
\begin{fact}
\label{fact:pushcond}
Pushing a demand $xy$ to $u$ maintains the cut condition in an Eulerian instance if and only
if there is no tight cut $\delta(S)$ that contains $u$ but none of
$x,y,v$.
\end{fact}

We call the preceding four operations {\em basic}, and we generally
assume throughout that our instance is reduced in that we cannot
apply any of these operations. In particular, we may generally assume
that $G$ is $2$-node connected.

\iffalse
Finally, while the last operation is not really critical in any setting, it is
useful for cleaning things up a bit.

\vspace*{.2cm}
\noindent
{\bf Contraction:}
If $z$ is a node not incident to any demand edge, and has only two
neighbours in $G$, say $u,v$. Then we may remove $z$, and add an edge
of capacity $c$ between $u,v$, where $c$ is the minimum of the
capacity of edges between $z$ and $u$, or $z$ and $v$.
\vspace*{.2cm}
\fi

\subsection{Series Parallel Instances}

A graph is {\em series-parallel} if it can be obtained from a single
edge graph $st$ by repeated application of two operations: series and
parallel operations.  A {\em parallel} operation on an edge $e$ in
graph $G=(V,E)$ consists of replacing $e$ by $k \geq 1$ new edges with
the same endpoints as $e$.
%We call the edges of the path children of $e$.
A {\em series} operation on an edge consists of replacing $e$ by a
path of length $k > 1$ between the same endpoints.
%We call the edges of the path children of $e$.
%In either case if $k=1$, the operation is called {\em trivial}, and
%{\em nontrivial} otherwise.
%Since $K_{2,3}$ is series-parallel, we
%know that the cut condition is not sufficient for the existence of
%multiflows in series-parallel graphs.
Series-parallel graphs can also
be characterized as graphs that do not contain $K_4$ as minor.

A {\em capacitated graph} refers to a graph where each edge also has
an associated positive integer capacity. For purposes of routing, any
such edge may be viewed a collection of parallel edges.  Conversely,
we may also choose to identify a collection of parallel edges as a
single {\em capacitated edge}.  In either case, for a pair of nodes
$u,v$, we refer to the {\em capacity} between them as the sum of the
capacities of edges with $u,v$ as endpoints.  For a pair of nodes
$u,v$ a {\em bridge} is either a (possibly capacitated) edge between
$u,v$ or it is a subgraph obtained from a connected component of
$G-\{u,v\}$ by adding back in $u,v$ with all edges between $u,v$ and
the component.  In the latter case, the bridge is {\em nontrivial}.  A
{\em strict cut} is a pair of nodes $u,v$ with at least $2$ nontrivial
bridges and at least $3$ bridges.

\begin{lemma}
\label{lem:decomp}
If $G$ is a $2$-node-connected series-parallel graph, then either it
is a capacitated ring, or it has a strict $2$-cut.
\end{lemma}
\begin{proof}
Suppose that $G$ has no strict $2$-cut.  Let $e_1,e_2, \ldots ,e_k$ be
the result of the first parallel operation on the original edge $st$.
Let $P_i=a_1,a_2, \ldots ,a_l$ be the path obtained after subdividing
$e_i$ if there is ever a nontrivial series operation applied to $e_i$
(where the $a_j$'s are the edges).  Since $G$ is $2$-node-connected,
$k \geq 2$.  If $k = 2$, then each edge of $P_1,P_2$ results in a
capacitated edge in $G$ and hence $G$ is a capacitated ring. If this
were not the case, then some $a_j$ has a parallel operation followed
by a series operation, and hence the ends of $a_j$ would form a strict
cut in $G$.  So suppose that $k \geq 3$. Clearly at most one $e_i$ is
subdivided, say $P_1$, or else $s,t$ is a strict cut.  Again, either
each edge of $P_1$ becomes a capacitated edge. Otherwise any series
operation to some $a_p$ results in its endpoints inducing a strict
$2$-cut.
\end{proof}

The following lemma is useful in applying the push operation (cf.
Fact~\ref{fact:pushcond}).

\begin{lemma}
  \label{lemma:centralrl}
  Let $u,v$ be a pair of nodes in a series parallel graph, and let $l,r$
  be a $2$-cut separating $u$ from $v$. Let $L$ be a central set
  containing $l$, but not $u$, $r$ and $v$; and let $R$ be a central set
  containing $r$, but not $u$, $l$ and $v$.  Then $L\setminus R$ and
  $R\setminus L$ are central.
\end{lemma}
\begin{proof}
    First, we prove that $V\setminus (R\setminus L) = (V\setminus R)\cup
  L$ is connected. This is the case since $(V\setminus R)$ and $L$ are
  each connected, and both contain $l$.

  Then, we prove that $R\setminus L$ is connected. Let $R_1$ be the
  connected component of $R\setminus L$ that contains $r$. Assume
  $R\setminus L$ contains another connected component $R_2$. The
  $2$-cut $l,r$ separates $u$ from $v$, so $V-l,r$ contains at least two
  connected components $C_u$ and $C_v$, containing $u$ and $v$
  respectively. Since $R_2$ contains neither $l$ nor $r$, it is
  entirely in a connected component $C$ of $V-l,r$. Assume without
  loss of generality that $C$ is not $C_u$. Since $R$ is connected,
  there is a path $P$ in $R$ from $R_1$ to $R_2$; choose $R_2$ so that $P$ is minimal.
  It follows that $P=(r_1,Q,r_2)$ where $r_i \in R_i$, and $Q$ is a subpath of $R \cap L$.
  Note that $P \cup R_2$ is disjoint from $C_u$.

 \parpic[r]{
  \begin{minipage}{3in}
    \begin{center}
      \includegraphics[height=2in]{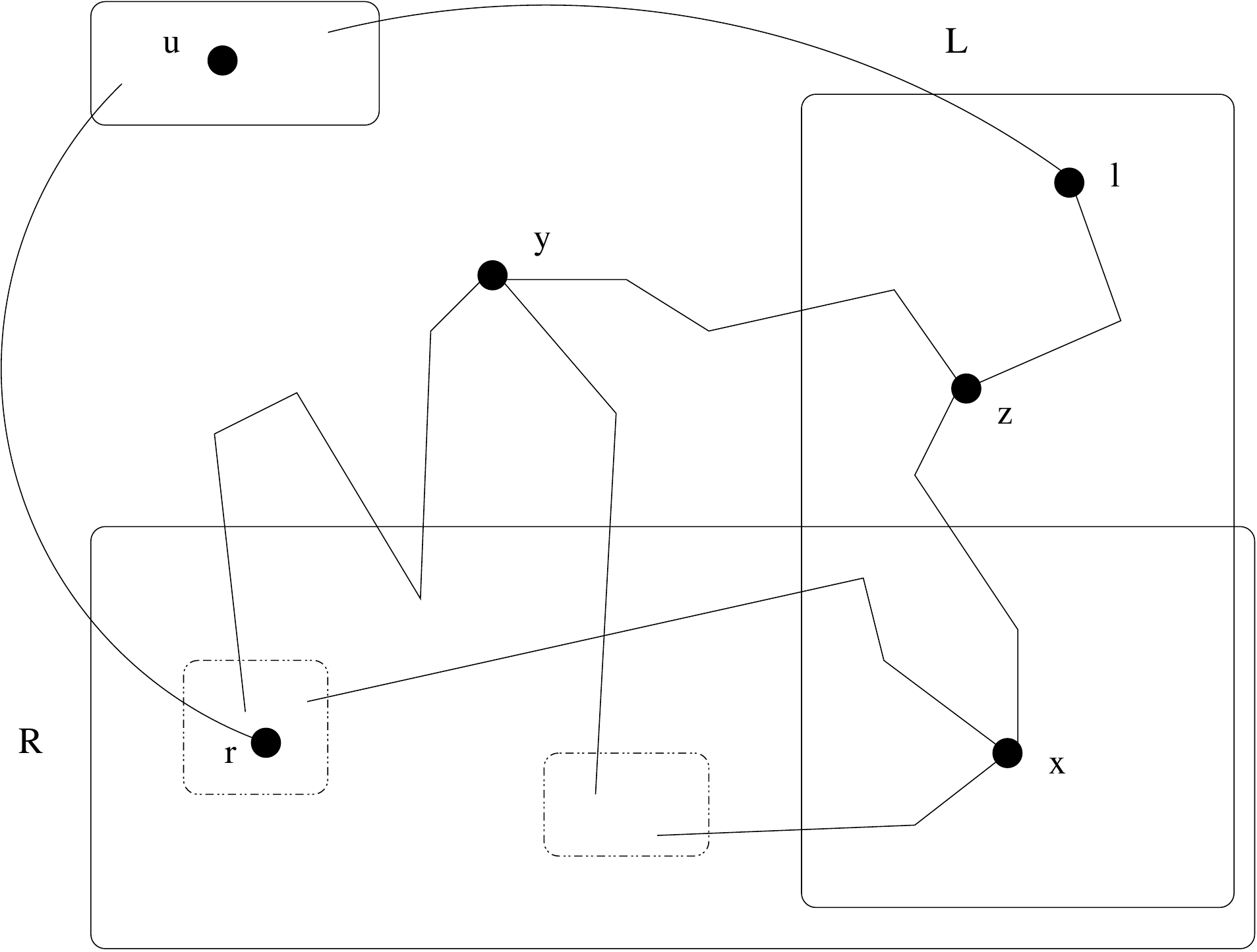}
    \end{center}
  \end{minipage}
}
  Since $R_1$ and $R_2$ are outside of $L$, there is a path $P'$ outside
  of $L$ connecting them. This path may have the form
  $a_1,Q_1,I_2,Q_2, \ldots ,Q_p,a_2$ where $a_i \in R_i$, each $Q_i$ is a path
  in $V \setminus (R \cup L)$ and each $I_j$ is a path in some component
  (other than $R_1,R_2$ in the graph induced by $R \setminus L$. Once again $P'$
  is disjoint from $C_u$, since it contains neither $l,r$ and hence is included in $C$.
  Since the internal nodes of $P'$ lie entirely in $V \setminus (R
  \cup L)$ and $P$'s lie in $R \cap L$, the paths are
  internally node-disjoint paths connecting $R_1,R_2$ in the graph $G-C_u$.

  Since $G[L]$ is connected, there is a path $P_1$ joining $l$ to an
  internal node of $P$ within $L$; since $G[V\setminus R]$ is
  connected, there is a path $P_2$ in this graph connecting $l$ to some
  internal node of $P'$. Choose these paths minimal, and let $x$ be the
  first point of intersection of $P_1$ with $P$, and $y$ be the
  first intersection of $P_2$  with $P'$. Once again, note that $P_1 \cup P_2
  \setminus l$ does not contain $l,r$ and hence is contained in
  $G-C_u$ since $y,x \in C$.  Choose $z \in P_1$ such that the subpath of $P_2$
  from $z$ to $y$ does not intersect $P_1$.  Clearly $z \in L \setminus R$ by construction.
  We now have constructed a $K_4-e$ minor (in fact a homeomorph) on the nodes $z,x,y$ and
  $r \in R_1$. This graph is also contained within
  $C$, except for $r$ and possibly $z$ if $z=l$. We may now extend
  this to a $K_4$ minor using $C_u$ and the subpath of $P_1$
  joining $l,z$.

  So $R\setminus L$ is connected, so it is central. $L\setminus R$ is
  central by the same argument.
\end{proof}
%The proof is included in the appendix (\ref{app:centralrl}).

\section{Instances where the Cut Condition is Sufficient for Routing}

\subsection{Fully Compliant Instances }
\label{sec:compliant}

Let $G$ be a series-parallel supply graph and $H$ a demand graph
defined on the same set of nodes. An edge $e$ of $H$ is \emph{fully
  compliant} if $G+e$ is also series-parallel.  An instance $G,H$ is
fully compliant if $G+e$ is series-parallel for each $e \in E(H)$.  We
note that $H$ itself may not be series-parallel in fully compliant
instances. For instance, we could take $G$ to be a ring and $H$ to be
the complete demand graph.

In this section we prove that fully compliant instances $G,H$ are
integrally routable if they satisfy the cut condition and $G+H$ is
Eulerian. This forms one base case in showing that compliant
instances (introduced in Section~\ref{sec:semicompliant}) are routable,
which in turn will yield our congestion 5 routing result for general
series-parallel instances.

We start with several technical lemmas.

%\subsubsection{Some fully compliant lemmas}

%\paragraph{Some fully compliant lemmas}

\begin{lemma}
 \label{lem:comp}
  Let $G$ be $2$-node-connected series-parallel graph. A demand edge $uv$
  is fully compliant if and only if there is an edge $uv$ in $G$, or $u,v$
  is a $2$-cut in $G$.
\end{lemma}
\begin{proof}
  Suppose that there is no edge $uv$ in $G$, and that $u,v$ is not a
  $2$-cut in $G$. Since $G$ is $2$-connected, it contains at least
  $2$ node-disjoint paths from $u$ to $v$. Since there is no edge
  $uv$, each of these paths contain at least one node apart from
  $u$ and $v$. And since $u,v$ is not a $2$-cut, there is a path
  connecting these two paths in $G-\{u,v\}$. Therefore, $uv$
  is not fully compliant, because $G+uv$ contains a $K_4$.

  Suppose $G+uv$ contains a $K_4$. Then there is no edge $uv$ in $G$,
  since it is series-parallel. Let $S_1$, $S_2$, $S_3$ and $S_4$ be
  the sets of nodes in $G+uv$ that form the $K_4$ minor. Without loss
  of generality $u$ and $v$ are in $S_1$ and $S_2$ respectively; they
  have to be in different sets for otherwise $G$ will have a $K_4$
  minor. Since $S_3$ and $S_4$ are adjacent, they are contained in
  the same connected component of $G-\{u,v\}$, say $C$. We claim that
  there cannot be another connected component $C'$ in $G-\{u,v\}$.
  For if it did, then since $G$ is $2$-node-connected, both $u$ and
  $v$ have an edge to $C'$. This implies that there is a path between
  $S_1$ and $S_2$ in $G$ that avoids $C$ (and hence $S_3$ and $S_4$);
  then $S_1,S_2,S_3,S_4$ would form a $K_4$ minor in $G$ which is
  impossible since $G$ is series-parallel. Therefore, there is no
  other connected component in $G-\{u,v\}$, and $u,v$ is not a
  $2$-cut.
\end{proof}

%
%In other words, an edge $uv$ is fully compliant if and only if it has at
%least two bridges.

%\Chandra{Does the above sentence make sense?}
% \Bruce - no, not by my definition of bridge. anyways its not needed.

\begin{lemma}
  \label{lem:2cut}
  Let $G$ be $2$-node-connected series-parallel graph. If an edge $uv$
  is not fully compliant, then there is a $2$-cut separating $u$ from $v$.
\end{lemma}
\begin{proof}
  Suppose an edge $uv$ is not fully compliant, and there is no $2$-cut
  separating $u$ from $v$. Then there are three node disjoint paths
  connecting $u$ and $v$. Since $uv$ is not fully compliant, there is no
  edge $uv$ and each of these paths contain at least one node apart
  from $u$ and $v$. By Lemma~\ref{lem:comp}, $u,v$ is not a $2$-cut,
  so there are paths connecting these paths in
  $G-\{u,v\}$. This creates a $K_4$ minor in $G$, which is
  impossible.
\end{proof}

%\subsubsection{A Fully Compliant $2$-Cut Reduction}
%\medskip
\paragraph{A $2$-Cut Reduction:}
A {\em partition} of $G$ is any pair of graphs $(G_1,G_2)$ such
that: (i) $V(G_1) \cap V(G_2) = \{u,v\}$, for distinct nodes $u,v$
(ii) $E(G)$ is the disjoint union of $E(G_1),E(G_2)$ and (iii)
$|V(G_i)| \geq 3$ for each $i$. Thus any $2$-cut admits possibly
several partitions, and we refer to any such as a {\em partition
for} $\{u,v\}$.  We say that a demand graph $H$ has no demands {\em
crossing} a partition for $u,v$, if $H$ can be written as a disjoint
union $H_1 \cup H_2$ where for $i=1,2$, $H_i$ is a subgraph of
$H[V(G_i)]$, the demand graph induced by one side of the partition.
Note that even if $G,H$ satisfy the cut condition, it may not be the
case that $G_i,H_i$ does. It is easily seen however that we may
always add some number $k_i$ of parallel edges between $u,v$ in each
$G_i$ so that $G_i,H_i$ does have the cut condition. For $i=1,2$ the
smallest such number is called the {\em deficit} of the reduced
instance $G_i,H_i$.  One easily checks that the deficit of at least
one of the reduced instances is at most $0$ if $G,H$ satisfies the
cut condition.

\begin{lemma}
\label{lem:reduce} Let $G,H$ satisfy the cut condition and let
$(G_1,G_2)$ be a partition for $2$-cut $\{u,v\}$ such that $H$ has
no demands crossing the partition. Let $k_i$ be the deficit of
$G_i,H_i$; without loss of generality $k_1 \geq 0 = k_2$. Let $H_2'$
be obtained by adding $k_1$ demand edges between $u,v$ in $H_2$. We
also let $H_1'=H_1$. Let $G'_1$ be obtained by adding $k_1$ supply
edges between $u,v$ to $G_1$; we also let $G_2'=G_2$. Then
$G'_i,H'_i$ satisfies the cut condition for $i=1,2$. Moreover, if
$G+H$ was Eulerian, then so is $G'_i+H'_i$ for $i=1,2$. Finally, if
there is an integral routing for each instance $G'_i,H'_i$, then
there is such a routing for $G,H$.
\end{lemma}
\begin{proof}
  First, let $S \subset V(G_1)$ which defines $G_1$'s deficit.  That
  is the number of demand edges in $\delta_{H_1}(S)$ is $k_1$ greater
  than $|\delta_{G_1}(S)|$. Without loss of generality, $u \in S$.
  Also, we must have $v \not \in S$, for otherwise $S \cup V(G_2)$
  violates the cut condition for $G,H$.  Clearly, $G_1',H_1$ now satisfies
  the cut condition. Next  suppose that $G_2,H'_2$
  does not obey the cut condition.  Then there exists some $S'$
  containing $u$ and not $v$, such that
  \[
  |\delta_{H_2}(S')| + k_1 > |\delta_{G_2}(S')|.
  \]
  But then $S \cup S'$ violates the cut condition for $G,H$.

  Let $G+H$ be Eulerian, and note that all nodes except possibly $u,v$ have even degree
  in $G'_i+H'_i$ (and in $G_i+H_i$). It is thus sufficient to show that $u,v$ also have
  even degree. Let $p$ be the parity of $u$ in $G_1+H_1$.
  Since any graph has an even number of
  odd-degree nodes, $v$ must also have parity $p$ in $G_1+H_1$.  Let
  $s=|\delta_{G_1}(S)|,d=|\delta_{H_1}(S)|$ and so $k_1 = d-s$.  Since
  $S$ separates $u,v$ we have that $d+s=|\delta_{G_1+H_1}(S)|$ has
  parity $p$ and hence $k_1 = d-s=d+s - 2s$ does as well.  That is, the
  deficit $k_1$ has parity $p$ and so $u$ and $v$ have even degree in $G'_1+H_1'$.
  This immediately implies the same for $G_2'+H_2'$.

  The last part of the lemma is immediate, since each demand from $H$ is routed
  in $G_1 \cup G_2$ together with the $k_1$ new supply edges in
  $G_1$. However, any such new edge was routed in $G_2$ together
  with the demands in $H_2$, so we can simulate the edge by routing
  through $G_2$.
\end{proof}

%\subsubsection{Proof that fully compliant instances are routable}

\paragraph{Fully Compliant instances are routable:}
Let $G,H$ satisfy the cut condition where $G$ is series parallel,
and each edge in $H$ is compliant (with $G$). We now show that if
$G+H$ is Eulerian, then there is an integral routing of $H$ in $G$.
The proof is algorithmic and proceeds by repeatedly applying the
reduction described above and those from Section~\ref{sec:basicops}.
In particular, at any point if there is a slack, parallel or $1$-cut
reduction we apply the appropriate operation. Thus we may assume
that $G$ is $2$-node connected, and that each demand or supply edge
lies in a tight cut, and no demand edge is parallel to a supply
edge.

If $G$ is a capacitated ring, then the result follows from the
Okamura-Seymour theorem. Otherwise by Lemma~\ref{lem:decomp}, there
is some strict cut $u,v$.  Thus $G$ has $3$ node-disjoint paths
$P_1,P_2,P_3$ between $u,v$ and so there is a partition $G_1,G_2$
for $u,v$ such that either $G_1$ or $G_2$ has $2$ node-disjoint
paths between $u,v$.  Without loss of generality, for $i=1,2$ $P_i$
is contained in $G_i$.  But then there could not be any demand edge
crossing the partition.  Since if $f \in E(H)$ has one end in
$G_1-\{u,v\}$ and the other in $G_2-\{u,v\}$, then we would have a
$K_4$ minor in $G+f$.  Thus we may decompose using
Lemma~\ref{lem:reduce} to produce two smaller instances and
inductively find routings for them. These two routings yield a
routing for $G,H$ by the last part of Lemma~\ref{lem:reduce}.

\subsection{Routable $K_{2m}$ Instances }

\parpic[r]{
  \begin{minipage}{3in}
    \begin{center}
      \includegraphics[height=1.5in]{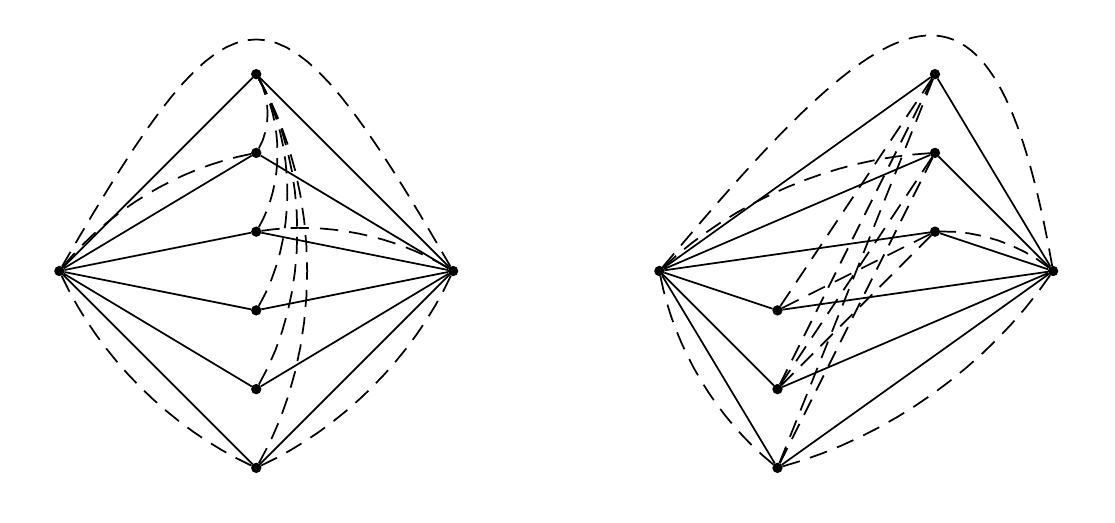}
    \end{center} \label{fig:trisource}
  \end{minipage}
}
A {\em $K_{2m}$-instance} consists of a supply graph $G=K_{2m}$ with a
$2$-cut $s,t$ and $m$ nodes $v_1,\ldots v_m$ of degree two, each
adjacent to $s,t$. We may possibly also have an edge between $s,t$. We
also have a demand graph $H=(V=V(G),F)$ on the same node set $V$, and
edge capacities $u$ on $G$'s edges.
A {\em path-bipartite
  instance} is one where the demands with both ends in the $m$ degree
$2$ nodes form a bipartite graph. One special case is a so-called
{\em tri-source} instance, where if $v_iv_j \in F$, then either $i$
or $j$ is $1$. The figure shows a tri-source and a path-bipartite
instance.

\begin{lemma}\label{lem:bipk2m}
If $G,H$ is a path-bipartite instance satisfying the cut condition
and $G+H$ is Eulerian, then there is an integral routing of $H$ in
$G$.
\end{lemma}
\begin{proof}
  If any demand edge $f$ is parallel to an edge $e \in G$, then we may
  delete $f$ and set $u_e = u_1-1$ to obtain a new Eulerian instance
  that one easily checks must satisfy the cut condition. Note that we
  may also assume that either there are no $st$ demands, or no $st$
  supply for otherwise we could reduce the instance. Thus we may
  assume that every demand edge either joins $s,t$ or is of the form
  $v_iv_j$ for some $i \neq j$. Suppose first that some node $v_i$
  does not define a tight cut. Consider the new instance obtained by
  adding a new supply edge between $s,t$ and remove one unit of
  capacity from each edge incident to $v_i$. The only central cuts
  whose supply is reduced is the cut induced by $v_i$; as this cut was
  not tight, it still satisfies the Eulerian and cut
  conditions. The new instance is also smaller in our measure and so
  we assume that each $\delta(v_i)$ is tight.

  Let $X,Y$ be a bipartition of the degree two nodes for the demands
  amongst them. Let $r_j = u_{v_jt}$ and $l_j = u_{sv_j}$ for each
  $j$. For a subset $S$ of the degree two nodes, we also let $r(S) =
  \sum_{v_i \in S} r_i$ (similarly for $l(S)$). Hence actual supply
  out of $\delta_G(X)$ is just $r(X)+l(X)$. We also let $d(i)$ denote
  the total demand out of $v_i$; hence we have that $d(i)=r_i + l_i$
  for each $i$. In particular, any subset of $X$ or of $Y$ is tight.
  Thus if $r(Y) < r(X)$, then $X \cup t$ is a violated cut and so
  $r(Y) \geq r(X)$. Similarly, $r(X) \geq r(Y)$. Thus $r(X)=r(Y)$ and
  the same reasoning shows that $l(X)=l(Y)$. Moreover, the above
  argument shows that there are no $st$ demands or else $X \cup t$ is
  a violated cut.

\parpic[r]{
  \begin{minipage}{7cm}
    \begin{center}
      \includegraphics[width=6cm]{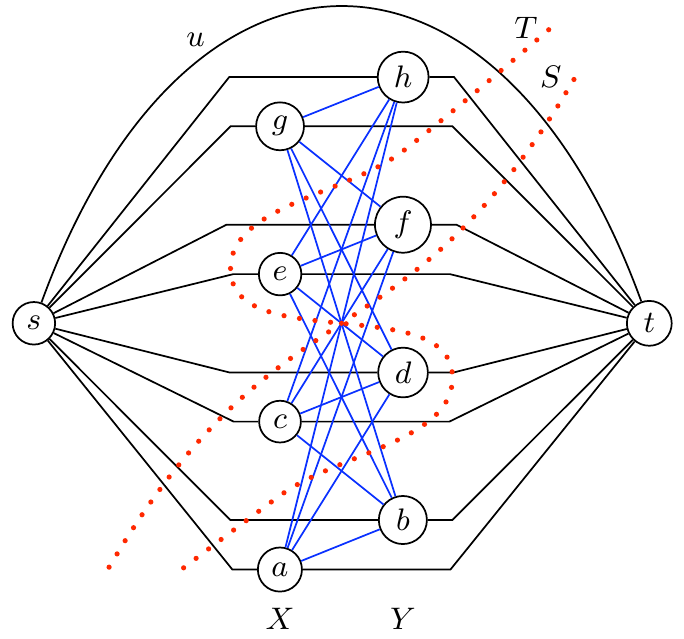}
    \end{center} \label{fig:pushingf}
  \end{minipage}
}
  Assume that there are two tight cuts $S$ and $T$ separating $s$ from
  $t$. We group the nodes of degree two and contract them in eight
  nodes, by inclusion in $X$ or $Y$, $S$ or $V\setminus S$, $T$ or
  $V\setminus T$.  Since nodes in $X$ (or $Y$) are not adjacent to each
  other, each node is still tight. The result is shown in the adjacent figure.
  %on   Figure~\ref{fig:bipartite}.

Some nodes and demand edges may not actually exist.

We denote by $l_x$ (resp. $r_x$) the capacity of the supply edge
between a node $x$ and $s$ (resp. $t$). We denote by $d_{xy}$ the
demand between nodes $x$ and $y$. We denote by $u$ the capacity of the
$st$ supply edge.

The cut induced by $S$ and $T$ are tight which implies:
$$
l_a+l_b+l_c+l_d+r_e+r_f+r_g+r_h - d_{af}-d_{ah}-d_{be}-d_{bg}-d_{cf}-d_{ch}-d_{de}-d_{dg}+u=0,
$$
$$
l_a+l_b+r_c+r_d+l_e+l_f+r_g+r_h - d_{ad}-d_{ah}-d_{bc}-d_{bg}-d_{cf}-d_{de}-d_{eh}-d_{fg}+u=0.
$$
The surplus of cuts induced by $\{s,d,f,g,h\}$ and $\{s,c,e,g,h\}$ must be positive:
$$
l_a+l_b+l_c+r_d+l_e+r_f+r_g+r_h - d_{ad}-d_{af}-d_{ah}-d_{cd}-d_{cf}-d_{ch}-d_{de}-d_{ef}-d_{eh}-d_{bg}+u\geq 0,
$$
$$
l_a+l_b+r_c+l_d+r_e+l_f+r_g+r_h - d_{ah}-d_{bc}-d_{cd}-d_{cf}-d_{be}-d_{de}-d_{ef}-d_{bg}-d_{dg}-d_{fg}+u\geq 0.
$$
If we add these two inequalities and subtract them by the two previous equalities, we get:
$$
-2d_{cd}-2d_{ef}\geq 0,
$$
and so $d_{cd}+d_{ef}= 0$. Therefore, for any two distinct $st$ cuts,
there is never a demand edge which is on the left of one and on the
right of the other.

So any demand edge can be pushed to $s$ or $t$. By induction, any
path-bipartite instance is routable.
\end{proof}
% The proof is included in the appendix (\ref{app:bipk2m}).

\subsection{$K_{2m}$ instances are cut-sufficient if and only if they have no odd $K_{2p}$ minor}
\label{subsec:odd}

While the preceding result on ``bipartite'' $K_{2m}$ instances will be
sufficient to deduce our congestion $5$ result later, we delve a little
more deeply into the structure of general $K_{2m}$ instances. This is
because they suggest a very appealing conjecture on cut-sufficient
series parallel instances (given below) and could very likely form a
basis for settling that conjecture.

Recall that a {\em $K_{2m}$-instance} is one whose supply graph is of
the form $G=K_{2m}$ with a $2$-cut $s,t$ and $m$ nodes $v_1,\ldots
v_m$ of degree two, each adjacent to $s,t$. We say an instance has an
{\em odd $K_{2p}$ minor} if it is possible to delete or contract
supply edges of the instance, and delete demand edges, until we get an
instance that is an \emph{odd $K_{2p}$}, that is, a $K_{2p}$ supply
graph, with $p$ odd, such that the two nodes of degree $p$ are
connected by a demand edge, and such that the $p$ nodes of degree $2$
are connected by an odd cycle of demand edges.  One easily sees that
an odd $K_{2p}$ instance is not cut-sufficient (refer to the
introduction for the definition of cut-sufficiency).  Indeed, setting
all demands and capacities to $1$ defines an instance that satisfies
both Eulerian and cut condition, but is not routable. In general, if
we only obtain this structure as a minor, we can think of assigning 0
demand (respectively capacity) to the deleted demand (resp. supply)
edges, and assign infinite capacity to any contracted supply
edge. Hence for any pair $(G,H)$, if it is cut-sufficient, then it
cannot contain an odd $K_{2p}$ minor. For general graphs this is not
sufficient (cf. examples in \cite{Schrijver_book}). However we
conjecture the following.
\begin{con}
Let $G$ be series parallel. Then a pair $(G,H)$ is cut-sufficient if
and only if it has no odd $K_{2p}$ minor.
\end{con}

\iffalse Note also that deleting a demand edge from an instance is
equivalent to setting its demand to zero, and that deleting,
respectively contract, a supply edge of an instance is equivalent to
setting its capacity to zero, respectively infinity. As a
consequence, any instance that has an odd $K_{2p}$ minor cannot be
cut-sufficient. \fi

Together with the results of the previous section, we now prove that
the conjecture holds for $K_{2m}$ instances. In other words, if such
an instance has no odd $K_{2p}$ minor, then it is routable. We prove a
slightly stronger statement.  Let $(G,H)$ be a $K_{2m}$ instance that
satisfies the cut condition, is Eulerian and has no odd $K_{2p}$
minor, then there is an integral routing for $H$ in $G$.

The graphs $G$ and $H$ are capacitated. The proof is by induction on
total demand.  In the following, we may assume that no demand edge is
parallel to a supply edge. Furthermore, we assume as in the proof of
Lemma~\ref{lem:bipk2m}, that any cut induced by a middle node $v_i$ is
tight. If the demand graph does not contain an odd cycle, then it is
bipartite and routable by Lemma~\ref{lem:bipk2m}. Let us therefore
assume that the demand graph contains at least one odd cycle. As the
graph should not have any odd $K_{2p}$ minor, there is no demand edge
from $s$ to $t$. We will subsequently prove that at least one demand
edge $f=v_iv_j$ can be pushed to $s$ or $t$, by assuming that none
can, and deducing a contradiction. Here, by pushing the edge
$f=v_iv_j$ to $s$ we mean that one unit of demand on $f$ can be
removed and new demand edges $sv_i$ and $sv_j$ with unit demand are
created (pushing to $t$ is similar).  Pushing allows us to complete
the proof by induction as follows.  Once we push a demand edge $f$ to
$s$ or $t$, the resulting instance could possibly have an odd $K_{2p}$
minor. However, the demand edges thus created are parallel to some
supply edges, and have only one unit of demand. If we reduce at once
the resulting new demand edges with the parallel supply edges, we get
a new $K_{2m}$ instance with supply and demand edges that all existed
in the original one, and so does not have an odd $K_{2p}$ minor
either. As the new instance has a smaller total demand, we can apply
induction.

Now we prove the existence of a demand edge that can be pushed.  We
call a demand edge a \emph{leaf demand edge} if it is the unique
demand edge incident to some $v_i$ node. The following lemma is useful
in showing that such demand edges can be pushed.

\begin{lemma}
\label{lemma:cuts}
Let $(G,H)$ be a $K_{2m}$-instance that satisfies the cut condition
and such that any cut induced by a $v_i$ is tight. Then any tight cut
separating $s$ from $t$ separates every node $v_i$ from at least one
of its adjacent nodes in the demand graph.
\end{lemma}
\begin{proof}
  Suppose there is some tight cut $S$ which separates a degree $2$ node $v$ from
  none of its adjacent nodes in the demand graph. We have
  $|\delta_G(S)|=|\delta_H(S)|$. Let us consider the cut $S'$ obtained by
  flipping the side of $v$ in the cut $S$. Then
  $|\delta_H(S')|=|\delta_H(S)|+ |\delta_H(v)|$, but
  $|\delta_G(S')|<|\delta_G(S)|+|\delta_G(v)|$, because $S$ separates $s$
  from $t$. Since by assumption $|\delta_H(v)|=|\delta_G(v)|$, the cut
  condition is not satisfied on $S'$, and so we have a contradiction.
\end{proof}

As a corollary of the previous lemma, any leaf demand edge in a
$K_{2m}$ instance can be pushed to $s$ (or to $t$), since no
tight cut can separate both its end points from $s$ (or from $t$).

Recall that the demand graph has an odd cycle.  We consider the
structure of the demand graph. We call a graph an
\emph{edge-$K_{2m}$-cycle} if it can be obtained from a cycle by
replacing some edges by $K_{2m}$ subgraphs.

\begin{lemma}
If a $K_{2m}$ instance does not have an odd $K_{2p}$ minor, and the demand
graph contains an odd cycle, but no leaf demand edge, then the demand graph
is either an edge-$K_{2m}$-cycle, or $K_4$.
\end{lemma}
\begin{proof}
  Let $C$ be the shortest odd cycle contained in the demand graph.  We
  claim that $V(C)$, the nodes in $C$, cover all demand edges.  For if
  a demand edge $v_iv_j$ is not covered, we can get an odd $K_{2p}$
  minor by contracting supply edges incident to $v_i$ and $v_j$.
  Further, if $v_iv_j$ is a demand edge with only one of $v_i,v_j$ on
  $C$, the edge is a leaf demand edge and hence can be pushed.
  Therefore we can assume any node $v_i$ not in $V(C)$ is adjacent to at 
  least two nodes of $C$ in the demand graph.

  First, suppose that $C$ is a triangle. Then any node $v_i \not \in
  V(C)$ adjacent to two nodes of $C$ creates another triangle. It is
  easy to see that any two nodes $v_i,v_j \not \in V(C)$ adjacent to
  two nodes of $C$ must be adjacent to the same two nodes of $C$,
  otherwise we can find a triangle and an additional edge not
  connected to it, which forms an odd $K_{2p}$ minor. Also, if
  $v_i,v_j \not \in V(C)$ are adjacent to the same two nodes of $C$,
  neither of $v_i,v_j$ is adjacent to the third node of $C$, because
  this would again form an odd $K_{2p}$ minor. So there are only two
  possibilities: Either the demand graph consists of many triangles
  all sharing the same edge of $C$, which is an edge-$K_{2m}$-cycle,
  or it is $K_4$.

  Now, suppose $C$ is a odd cycle of length at least $5$.  Suppose a
  node $v \not \in V(C)$ is adjacent to two distinct nodes $v_1, v_2
  \in V(C)$ in the demand graph. Then $v_1v_2$ is not an edge of the
  demand graph for otherwise $v$ along with $v_1,v_2$ would form a
  triangle in the demand graph contradicting the fact that $C$ is the
  shortest odd cycle. Suppose the shortest path in $C$ connecting
  $v_1$ and $v_2$ has length $3$ or more. Then the edges $vv_1$ and
  $vv_2$ create a shortcut for $C$, which creates shorter cycles with
  each of the two paths in $C$ from $v_1$ to $v_2$, one of them odd,
  contradicting the choice of $C$.  Therefore, the shortest path in
  $C$ connecting $v_1$ and $v_2$ must have length $2$. This
  also implies that $v$ has degree $2$, because three nodes in $C$
  cannot all be at distance $2$ of each other.

  So any node $v \not \in V(C)$ forms a bridge of length $2$ between
  two nodes $v_1$ $v_2$ of $C$, with $v_1$ and $v_2$ being at distance
  $2$ in $C$. Let $v'$ be the middle node on the length $2$ path
  between $v_1$ and $v_2$ in $C$. We claim that $v'$ has degree $2$ in
  the demand graph. To see this, note that the cycle $C'$ obtained by
  replacing the path $v_1,v',v_2$ in $C$ by the path $v_1,v,v_2$ has
  the same length as $C$ and $v' \not \in V(C')$. By the argument in
  the previous paragraph, $v'$ has degree $2$. It follows that $v_1$
  and $v_2$ form a $2$-cut for the demand graph, and all but one of
  the connected components obtained by removing $v_1$ and $v_2$ are
  singletons. As any node $v \not \in C$ is adjacent to exactly two
  nodes $v_1$ and $v_2$ for which this is true, the demand graph is an
  edge-$K_{2m}$-cycle.
\end{proof}

It is known that all demand graphs on less than five nodes are
routable in any graph that satisfies the cut condition for it (see
\cite{Schrijver_book}), and hence $K_4$ is routable. So if we assume
that a $K_{2m}$ instance has no odd $K_{2p}$ minor, but that no demand
edge can be pushed, then the demand graph must be an
edge-$K_{2m}$-cycle.  We now show a property of the capacity and
demand vectors of such instances.

\begin{lemma}\label{lemma:bracket}
  Let $(G,H)$ be a $K_{2m}$ in which $H$ is an edge-$K_{2m}$-cycle, no
  demand edge can be pushed, and any cut induced by a $v_i$ is
  tight. Then for any node $v$ that has degree $2$ in both the demand
  and supply graph, the capacity and demand of the four edges incident
  to $v$ in both graphs is the same.
\end{lemma}
\begin{proof}
  Let $v$ be a node of degree $2$ in the demand graph, connected to
  $v_1$ and $v_2$. By assumption, no edge of the demand graph can be
  pushed. So there are tight cuts $S_1$, $T_1$ separating $vv_1$ from
  $s$ and $t$ respectively. Similarly, there are tight cuts $S_2$,
  $T_2$ separating $vv_2$ from $s$ and $t$ respectively. By
  Lemma~\ref{lemma:cuts}, the cuts $S_1$ and $T_1$ must both separate
  $vv_1$ from $v_2$. Similarly, $S_2$ and $T_2$ must separate $vv_2$
  from $v_1$. Let us consider $S_1$ and $T_1$, and flip the side of
  $v$ in both. The total capacity $|\delta_G(S_1)|+|\delta_G(T_1)|$ of
  both cuts is the same as before, but the total demand is modified by
  replacing the demand of $vv_2$ by that of $vv_1$. Since $S_1$ and
  $T_1$ were already tight, this means the demand of $vv_1$ is no
  greater than that of $vv_2$.  Symmetrically, starting from $S_2$ and
  $T_2$, we prove that the demand of $vv_2$ is no greater than that of
  $vv_1$, and so the demand on both edges is the same. By repeating
  the argument on the pair of cuts $S_1$ and $S_2$, and the pair $T_1$
  and $T_2$, we prove that the capacity of the edge $sv$ is the same
  as that of $vt$.  Since the cut induced by $v$ is tight, the
  capacities and demands are all equal.
\end{proof}

In the following, we call a node $v$ \emph{bracketed} by $v_1$ and
$v_2$ if it is of degree $2$ in the demand graph, and connected to
$v_1$ and $v_2$.  The previous lemma has the following easy corollary:
\begin{corollary}
  Let $(G,H)$ be a $K_{2m}$ in which $H$ is an edge-$K_{2m}$-cycle, no
  demand edge can be pushed, and any cut induced by a $v_i$ is tight.
  Let $v$ be a node bracketed by $v_1$ and $v_2$. Let $S$ be a tight
  cut separating $s$ from $t$. Then either $S$ separates $v$ from both
  $v_1$ and $v_2$, or $S$ separates $v_1$ from $v_2$, in which case
  the cut obtained by flipping the side of $v$ in $S$ is also tight.
\end{corollary}
\begin{proof}
  If $S$ does not separate $v_1$ from $v_2$, then it must separate them
  from any node they bracket, because otherwise we get a tighter cut by
  flipping the bracketed node. If $S$ does separate $v_1$ from $v_2$,
  then by Lemma~\ref{lemma:bracket}, flipping $v$ does not change
  the surplus of the cut.
\end{proof}

The above corollary implies that for any tight cut separating $s$ from
$t$, there is another tight cut such that all nodes bracketed by a
pair $v_1$ and $v_2$ are on the same side of the cut.  Let us now
consider any demand edge. If it cannot be pushed then there are tight
cuts $S$ and $T$ separating the end points of the demand edge from $s$
and $t$ respectively. We flip all bracketed nodes so that they are all
on the same side of $S$, and all on the same side of $T$. We can now
consider all nodes bracketed by a same pair as a single node, since
they are all on the same side of $S$ and $T$. If we do so for all
bracketed nodes, we reduce the demand graph to a cycle, which is
routable, and hence contradicts the assumption that $S$ and $T$ are
tight.

\subsection{Compliant Instances}
\label{sec:semicompliant}

We define in this section the notion of compliant instance, and
prove any such instance is routable if it is Eulerian and satisfies
the cut condition. This is the main technical contribution of
the paper.
%In the next section, we  further prove that any
%instance can be made to be compliant by multiplying capacities
%by four.
Recall that a demand edge $e$ is called fully compliant if $G+e$ is
series-parallel. If $G$ is a series-parallel graph created from the
edge $st$, we orient the edges of $G$ by orienting the initial $st$
and extending it naturally through series and parallel operations.
We abuse notation and use $G$ to refer to both the undirected
graph and the oriented digraph.

In the resulting digraph, $s$ is a unique source, and $t$ is a
unique sink; it is easy to see that this property is not lost by any
series or parallel operation. The graph is also acyclic, because we
can build an acyclic order starting from one for the $st$ edge and
extending it through the sequence of series and parallel operations.
As a consequence, any directed path can be extended to an $st$ path,
because we can always add an edge at the beginning until it starts
from $s$, and at the end until it ends at $t$.

We call a demand edge \emph{compliant} if there is a directed
path in $G$ connecting its endpoints. An instance is
\emph{compliant} if all edges are compliant or fully compliant.
(It is easy to show that if the $s,t$ cut has three or
more bridges, then in fact any fully compliant demand edge is also compliant.)

\iffalse
%BRUCE:  I shortened to the above case since the full blown thing didnt add too much i thought.
is an edge such that any
  $st$ path contains either none or both of its endpoints. In that case,
  any fully compliant edge is therefore compliant. We could add this here,
  or do we leave it out?}
\fi

\begin{theorem}\label{thm:semi}
  Let $G$ be a series-parallel graph. Further let $G,H$ be a
  compliant instance with $G+H$ Eulerian. If $G,H$ satisfies the
  cut condition, then $H$ has an integral routing in $G$.
\end{theorem}

We start by two technical
lemmas on oriented series-parallel graphs and on
compliant demand edges.

\begin{lemma}
\label{lemma:crosstwice}
  Any directed path in $G$ crosses at most twice the cut defined by a
  central set.
\end{lemma}

\begin{proof}
  Suppose there is a directed path $P$ that crosses at least three
  times the cut $\delta(S)$ defined by a central set $S$. We extend
  $P$ to an $st$ path $P'$, and denote by $e_1$, $e_2$ and $e_3$ the
  first three edges of $P'$ crossing $\delta(S)$. $P'$ can be
  decomposed as follows: $P'=(P_1,e_1,P_2,e_2,P_3,e_3,P_4)$ with $P_1$
  starting from $s$ and $P_4$ ending at $t$. Without loss of
  generality, we can assume $P_1$ and $P_3$ do not intersect $S$,
  $P_2$ is contained in $S$, and $P_4$ is partially in $S$.
\parpic[r]{
  \begin{minipage}{2.5in}
    \begin{center}
      \includegraphics[width=2.5in]{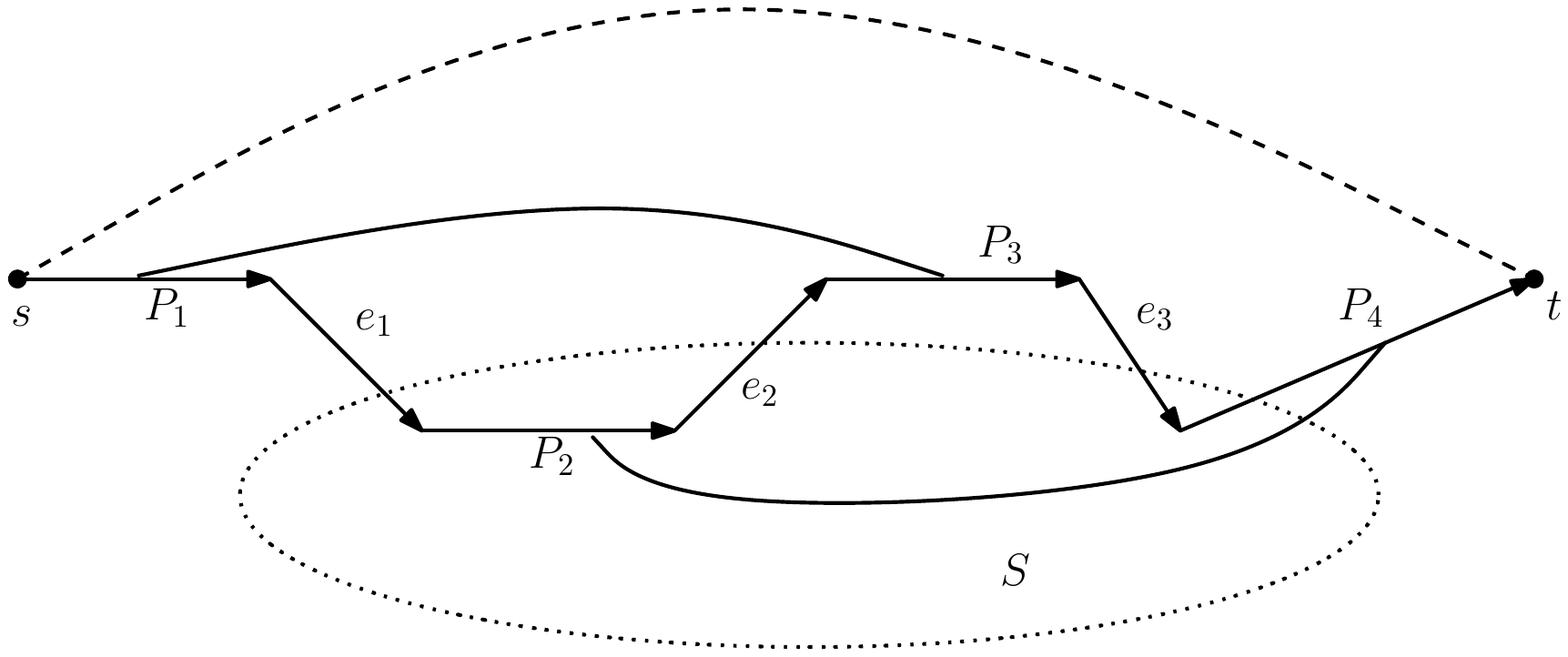}
    \end{center} \label{fig:central}
%\centerline{\input{comp.pstex_t}}
%\caption{A series-parallel instance satisfying the cut condition,
%but with no multiflow (fractional  routing).}
  \end{minipage}
}
We show
  these four parts form a $K_4$, which is a contradiction. Indeed,
  $P_1$ and $P_3$ are connected by a path which does not cross
  $\delta(S)$, because $S$ is central. Similarly, $P_2$ and $P_4$ are
  connected by a path contained in $S$. Since $S$ is a series-parallel
  graph created from an $st$ edge, adding an $st$ edge should not
  create a $K_4$ minor. However, this would clearly be the case here,
  and so we have a contradiction.
\end{proof}

\begin{lemma}
\label{lemma:technical}
Let $G$ be a $2$-connected series parallel graph obtained from an edge
$st$. Let $uv$ be a compliant edge that is not fully compliant. Then
there is a directed $s$-$t$ path $P$ that contains $u,v$ and in
addition, there is a $2$-cut $l,r$ in $G$ that separates $u,v$ such
that $l,r$ lie on $P$. Moreover, if $P$ traverses $s,l,u,r,v,t$ in
that order, then $l,r$ can be chosen to separate $u$ from both $s$ and
$t$ (and symmetrically, if $P$ traverses $s,u,r,v,l,t$ we can separate
$v$ from $s,t$).
\end{lemma}
\begin{proof}
  Since $uv$ is compliant, there is a directed path which without
  loss of generality traverses $u$ and then $v$. This can be extended
  to a directed $s$-$t$ path $P$ traversing $s,u,v,t$ in that order.
  Since $uv$ is not fully compliant, there is by Lemma~\ref{lem:2cut} a
  $2$-cut $l,r$ in $G$ separating $u$ from $v$.  One node of the
  $2$-cut, say $r$, has to be on the path from $u$ to $v$. So $P$
  traverses $s-u-r-v-t$ in that order.  Suppose $l$ is not on $P$;
  then $l,r$ separates $s$ from $t$ (otherwise, there would be a path
  $u-s-t-v$ in $G-\{l,r\}$). Since $G$ is $2$-node connected, there
  are two nodes-disjoint paths from $s$ to $t$, one containing $l$ and
  the other containing $r$.  This implies that $s,t$ is also a $2$-cut
  separating $l$ from $r$. If this were not the case, then $G-\{s,t\}$
  would contain a path joining $l,r$.  But then $G+st$ would clearly
  contain a $K_4$ minor, a contradiction. We now claim that $s,r$ is a
  $2$-cut that separates $u,v$, and hence choosing $l=s$ gives the
  desired cut. Suppose not, then there is a $u$-$v$ path $Q$ in
  $G-\{s,r\}$; this path $Q$ necessarily has to contain $l$ since
  $l,r$ is a $2$-cut for $u,v$. Let $Q_1$ be the sub-path of $Q$ from
  $u$ to $l$. If $Q_1$ contains $t$, then we have a $u-v$ path that
  avoids $l,r$, by following $Q_1$ and then the portion of $P$ between
  $t$ and $v$, a contradiction.  Otherwise, $Q_1$ combined with the
  portion of $P$ from $u$ to $r$ is an $l-r$ path that avoids $s,t$,
  again a contradiction. Hence we can assume that $l$ also lies on $P$
  as claimed.

  For the second part, assume that $P$ traverses $s,l,u,r,v,t$ in
  order.  Clearly, there cannot be a path from $u$ to $t$ in
  $G'=G-\{l,r\}$ for otherwise it can be combined with $P$ to find a $u-v$ path
  that avoids
  $l,r$. Suppose $u$ is connected to $s\neq l$ in $G'$. Since $G$ is
  $2$-node-connected, there are two oriented nodes-disjoint $s$-$t$
  paths in $G$ and one of them avoids $l,r$; otherwise one contains $l$ and the other $r$.
  Hence these two paths
  along with the portion of $P$ from $l$ to $r$ yield a $K_4$ in
  $G+st$. Thus $s,t$ are connected in $G'$ and hence if $u$ can reach
  $s$ in $G'$ it can reach $t$, and hence $v$ as well, again contradicting that $l,r$
  is a $2$-cut for $u,v$.
  \end{proof}

We are now ready to prove Theorem~\ref{thm:semi}.

\begin{proof}
 Let $uv$ be a demand edge which is compliant, but
  not fully compliant. We show that we can push $uv$ into a series
  of fully compliant demand edges, maintaining the hypotheses for the new instance.

  By Lemma~\ref{lemma:technical} we have, without loss of generality, a directed $s$-$t$ path $P$ that
  traverses nodes $s,l,u,r,v,t$ in that order (possibly $s=l$)
  where the component of $G-l,r$ containing $u$  does not contain $v,s$ or $t$.

  \begin{figure}[htb]
    \begin{center}
      \includegraphics[height=2in]{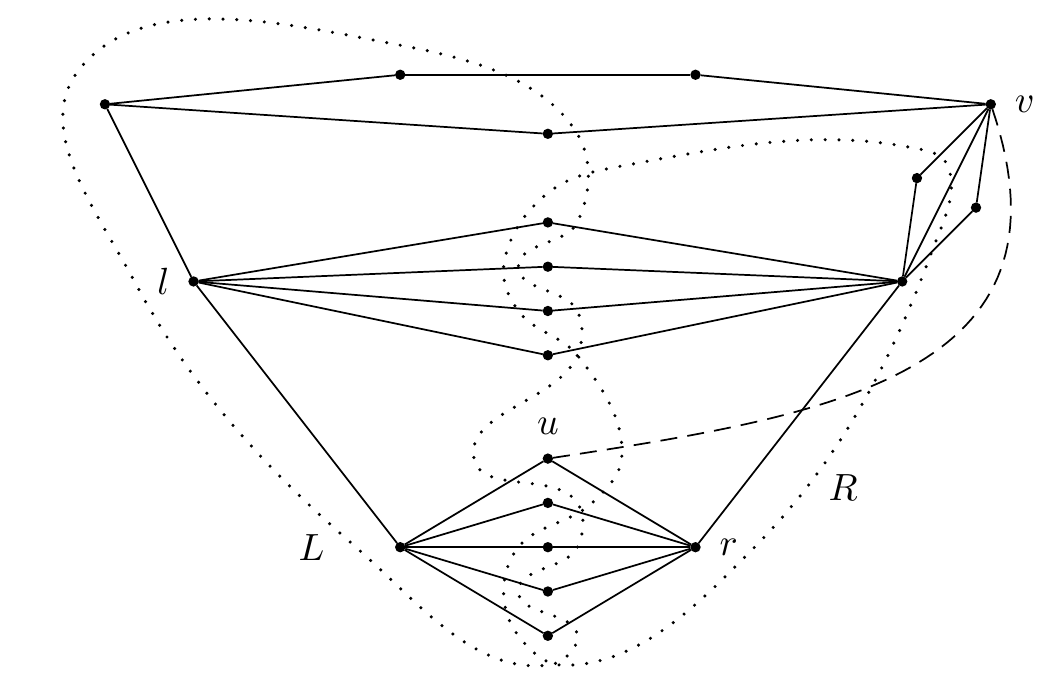}
    \end{center}
    \caption{$u$ and $v$ are separated by the $2$-cut $l,r$, and the demand $uv$ cannot be pushed to $l$ or $r$.}
    \label{fig:contraction}
  \end{figure}

  We show that we can push the edge $uv$ to $l$ or $r$. Suppose this
  is not the case, then since $l,r$ is a $2$-cut, by
  Fact~\ref{fact:pushcond} there is a (central) tight cut $\delta(L)$
  separating $l$ from $u,r$ and $v$, and a (central) tight cut
  $\delta(R)$ separating $r$ from $u,l$ and $v$. We show that this is
  not possible, by contracting the graph into a tri-source or a
  path-bipartite instance which again contains tight cuts
  corresponding to $\delta(L)$ and $\delta(R)$.  Since we know that
  such instances are routable, these tight cuts could not exist.

  %Since $R$ is central, there is a path from $u$ to $v$ in $V\setminus
  %R$, which goes through $l$; since $l,r$ is a $2$-cut for $u,v$, this
  %path goes through $l$ and hence

  Recall that $L \setminus R$ contains $l$ but not $r$,
  and $R \setminus L$ contains $r$ but not $l$. Hence we can contract
  $L\setminus R$ and $R\setminus L$ and label the  nodes $l'$ and $r'$
  respectively.  (By Lemma~\ref{lemma:centralrl}, this is actually a
  contraction on connected subgraphs, although it is not critical at
  this point in the argument that we have a minor as such.)
  Denote the resulting instance (after removing loops) by $G^*,H^*$.

  Since $\{l,r\}$ is a $2$-cut separating $u,v$ we have that the graph
  induced by the nodes $V\setminus(L\Delta R)$ has at least $2$
  components (one containing $u$ and one containing $v$).  Let $C_i$
  denote the components in this graph and let $X_i = R \cap L \cap
  C_i, Y_i = C_i \cap (V \setminus (R \cup L))$. Consider any
  component $K$ of $G[Y_i]$.  Since $G-L$ is connected, there is a
  path from $K$ to $R \setminus L$ in the graph induced by $K \cup (R
  \setminus L)$. Similarly, there is a path from $K$ to $L \setminus
  R$ in $G[K \cup (L \setminus R)]$, since $G-R$ is connected.
  Analogously, if $K$ is a component of $G[X_i]$, then since $G[R]$ is
  connected, there must be a path from $K$ to $R \setminus L$ in $G[K
  \cup (R \setminus L)]$.  Similarly there is a path from $K$ to $L
  \setminus R$.  Now if both $X_i,Y_i$ are nonempty, then we may
  choose a pair of components $K,K'$ from $X_i,Y_i$ respectively, and
  a path joining them in $C_i$, so that we can form $K_4-e$ in the
  minor whose nodes are $l',r',K,K'$. This together with a path
  through some other $C_j$ yields a $K_4$. Hence for each $i$, at most
  one of $X_i$ or $Y_i$ is nonempty. Moreover, if we shrink each $C_i$
  to a node $c_i$, then we have edges $c_il',c_ir'$. The shrunken
  graph is therefore of type $K_{2m}$, with possibly an edge from $l'$
  to $r'$. Since each $c_i$ was either of ``type'' $R \cap L$, or type
  $V\setminus(R \cup L)$, in the shrunken graph we can identify tight
  cuts induced by sets $L',R'$ associated with our original pair
  $L,R$.

  We next claim that neither $s$ nor $t$ is in $R$ (and hence
  $R\setminus L$); to see this, note that the $s$-$t$ path $P$ goes
  successively through $s,u,r,v,t$ and hence if either $s$ or $t$ is
  in $R$, then $P$ would cross $\delta(R)$ three or more times.

  We now examine two cases based on whether $t$ is in $L$ or not.
  In each case we examine the structure of the demand edges in
  the shrunken graph.

%  Since in $G$, the path $P$ goes successively through $u$, $r$ and
%  $v$, $s$ and $t$ are not in $R$, because then $P$ would cross
%  $\delta(R)$ three times.
  \begin{enumerate}
  \item $t$ is not in $L$. So $t$ is contained in some $C_i$, say
    $C_1$. Suppose first that $s$ lies in some $C_i$ different from
    $C_1$.  By Lemma~\ref{lemma:technical}, $l,r$ separates $u$ from
    $s,t$; so $u \not\in C_i \cup C_1$.  But then adding $st$ (i.e.,
    $c_ic_1$) to the shrunken graph (which maintains the
    series-parallel property), would create a $K_4$ on the nodes
    $l',r',c_i,c_1$, by considering the $l'-r'$ path through the
    component $C_j$ containing $u$. Hence we may assume that $s$ is
    either in $L$ or in $C_1$.

   Consider next some compliant
    demand edge $u'v'$ in the shrunken instance, and suppose that neither of its endpoints lie in
    $C_1$;
    say the endpoints are in $C_i,C_j$.
    Consider the directed $s$-$t$ path $P'$  associated with this
    demand where $P'$   traverses $s,u',v',t$ in that order.  In
    the shrunken graph it must cross the cuts induced by
    $C_1,C_i,C_j$ at least $5$ times. % BRUCE picture would probably be nice
    Since every edge in any of these cuts
    lies in either $\delta(L')$ or $\delta(R')$, one of these two cuts
    is crossed three times, a contradiction. Hence, any compliant edge that remains in the shrunken graph,
    must have one endpoint in $C_1$.  Thus the shrunken graph is a
    tri-source instance where $l'$, $r'$ and $c_1$ are the three
    sources.

  \item $t$ is in $L$. We claim that $s$ is also in $L$; $P$ goes
  successively through $l$, $u$ and $t$, and hence if $s$ is not in
  $L$, $P$ would cross $\delta(L)$ three times.

  Recall that any directed path can be extended into a path starting
  at $s$ and ending at $t$, and that this path should cross the cut of
  any central set at most twice by Lemma~\ref{lemma:crosstwice}.
  We also use Lemma~\ref{lemma:centralrl} to obtain that both $R \setminus L,L\setminus R$
  are central sets. Thus since $s$ and $t$ are
  both in $L\setminus R$, there are two types of
  directed $s$-$t$ paths. The first type does not ever enter $R \setminus L$.
  If such a path ever leaves  $L\setminus R$, then it must
  enter some $C_i$. It must  then leave $C_i$ to reach $t$, and
  hence it enters $L\setminus R$ again. Thus the path has crossed the cut of
  $L\setminus R$ twice, and hence it can not leave it again. The second
  type of path may
  traverse $R\setminus L$. This is the only type that can traverse
  more than one $C_i$. We claim that a directed $s$-$t$ path of the
  second
  type goes from $L\setminus R$ to $R\setminus L$, traversing at most
  one $C_i$ on the way, then goes back from $R\setminus L$ to
  $L\setminus R$, traversing at most one $C_i$ on the way.
  This follows since any $s$-$t$ directed path cannot cross
  $\delta(R \setminus L)$ more than twice, so it can enter at most
  once, and leave at most once.

  Therefore, a directed $s$-$t$ path of the second type leaves $L\setminus
  R$, possibly traverses a $C_i$, enters $R\setminus L$, leaves
  $R\setminus L$, possibly traverses a $C_i$, and enters $L\setminus
  R$ in that order. It cannot leave $L\setminus R$ again after, as it
  has crossed its cut twice. As a consequence, no directed path can
  traverse more than two $C_i$'s.

  We now show that for any $C_i$, the arcs (oriented edges of $G$)
  connecting $C_i$ to $R\setminus L$ are all oriented in the same
  direction. Assume there is a $C_i$, say $C_1$, with an arc $e_1$
  entering and an arc $e_2$ leaving $R\setminus L$. We can extend
  $e_1$ into a directed path $P_1$ starting with $s$. Since $P_1$ ends
  by entering $R\setminus L$ with $e_1$, it did not leave $R\setminus
  L$ before, so it entered $C_1$ from $L\setminus R$. We can also
  extend $e_2$ into a directed path $P_2$ ending at $t$. Since $P_2$ starts
  by leaving $R\setminus L$ with $e_2$, it can not enter $R\setminus
  L$ again, so it leaves $C_1$ for $L\setminus R$.

\parpic[r]{
  \begin{minipage}{2in}
    \begin{center}
      \includegraphics[height=1.5in]{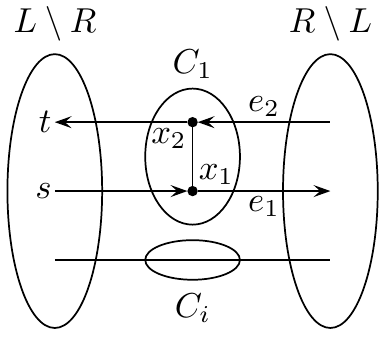}
    \end{center}
  \end{minipage}
}

  %So these two paths traverse $C_1$, $P_1$ from $s$ to $R\setminus L$
  %and $P_2$ from $R\setminus L$ to $t$.
We now show $P_1$ and $P_2$ are node-disjoint inside $C_1$. Suppose
that it is not the case, and that there is some $x\in C_1$ contained
both in $P_1$ and $P_2$. Then we can create a directed path leaving
$R\setminus L$ then entering it again, by following $P_2$ from $e_2$
to $x$, then following $P_1$ from $x$ to $e_1$. The directed path is
simple, since the graph is acyclic. We can extend it to a directed
$s$-$t$ path which then contradicts our previous claims.  Hence $P_1$
and $P_2$ are node-disjoint inside $C_1$. But since $C_1$ is
connected, these paths are connected by some path. That is, for any
$x_1\in C_1$ on $P_1$ and any $x_2\in C_1$ on $P_2$, there is a path
inside $C_1$ connecting $x_1$ to $x_2$.  Also, since $u$ and $v$ are
not in the same $C_i$, there is at least one other $C_i$ linking
$L\setminus R$ and $R\setminus L$. This now creates a $K_4$ in the
graph (with $l',r'$ shrunken) on the nodes $l',r',x_1,x_2$, a
contradiction.

  Thus for any $C_i$, the arcs connecting $C_i$ to
  $R\setminus L$ are all oriented in the same direction. Therefore
  the
  $C_i$'s are partitioned into two types: {\em out} components are those with arcs going into
  $R\setminus L$, and {\em in} components  with arcs entering from $R\setminus
  L$.  It follows that  any directed path traversing two $C_i$'s traverses one of
  each type.  Any compliant edge $u'v'$ that remains in the shrunken
  graph that is not incident to $l'$ or $r'$, must admit a directed $u'v'$ with its
  endpoints
  in distinct components. Hence exactly one of $u',v'$ lies in an in component,
  and the other lies in an out component. Thus the shrunken graph is a
  path-bipartite instance.
  \end{enumerate}

In either case, the instance $G^*,H^*$ obtained is routable by
Lemma~\ref{lem:bipk2m}. Therefore, the cuts $\delta(L)$ and
$\delta(R)$ could not have  both been tight, and so we can push $uv$
to either $l$ or $r$. By induction, we can push any compliant
edge into a series of fully compliant edges.
\end{proof}

\section{Integral Routing with Congestion}

\subsection{Routing in Congestion 5 in Series Parallel Graphs}
\label{sec:spcong5}

\begin{theorem}
Suppose $G,H$ is a series-parallel instance satisfying the cut
condition. Then $H$ has an integral routing with edge congestion
$5$.
\end{theorem}
\begin{proof}
  By the result of \cite{LeeR07}, there is a fractional routing $f$ of
  $H$ with congestion $2$. For any demand edge $xy$, suppose that
  $s',t'$ induce the highest level (with respect to the decomposition
  of $G$ starting from $st$) 2-cut separating $x,y$.  Then at least
  half of any fractional flow for $xy$ has to go either via $s'$ or
  $t'$.  Without loss of generality, assume it is $s'$. We push the
  $xy$ demand edge to $s'$, i.e., we create demand edges $xs',ys'$ and
  remove $xy$.  We do this simultaneously for all demand edges - we
  are pushing demands based on the fractional flow $f$. The new
  instance is compliant. Let us call $H'$ the new demand graph. By
  construction, there is a feasible fractional flow of $1/2$ of each
  demand from $H'$ in $2G$. This implies that $4G$ satisfies the cut condition
  for $H'$.  In order to make the graph Eulerian, we can add a
  $T$-join, where $T$ is the set of odd degree nodes in $4G+H'$. Since
  we can assume that $G$ is connected by previous reductions, we may
  choose such a $T$-join as a subset of $E(G)$. It follows that we can
  create an Eulerian, compliant instance $G',H'$ that satisfies the
  cut condition, and $G'$ is a subgraph of $5G$.  Hence by
  Theorem~\ref{thm:semi}, we may integrally route $H'$, and hence $H$
  in the graph $5G$.
\end{proof}

We believe that the above result can be strengthened substantially
and postulate the following:
\begin{con}\label{con:SP}
  Let $G,H$ satisfy the cut condition where $G$ is series-parallel and
  $G+H$ is Eulerian. Then there is a congestion $2$ integral routing
  for $H$.
\end{con}

\subsection{Rerouting Lemma from \cite{treewidth}}
\label{sec:rerouting}
We state the rerouting lemma that we referred to in the
introduction. It is useful to refer to the informal version we described
earlier.  Let $D$ be a demand matrix in a graph $G$ and let $f : V
\rightarrow V$ be a mapping. We define a demand matrix $D_f$ as
follows:

\[
D_f(xy) = \sum_{uv: f(u) = x, f(v) = y} D(uv).
\]
In other words the demand $D(uv)$ for a pair of nodes $uv$ is
transferred in $D_f$ to the pair $f(u)f(v)$. Thus the total demand
transferred from $u$ to $f(u)$ is $\sum_v D(uv)$. We define another
demand matrix $D'_f$ which essentially asks that each node $u$ can
send this amount of flow to $f(u)$.

\[
D'_f(uf(u)) = \sum_{v} D(uv).
\]

\begin{proposition}
  \label{prop:compose-routing}
  If $D'_f$ is (integrally) routable in $G$ with congestion $a$, and
  $D_f$ is (integrally) routable in $G$ with congestion $b$, then $D$
  is (integrally) routable with congestion $a+b$ in $G$.
\end{proposition}

We need a cut condition given by the simple lemma below.
%For any $\alpha \in \cR^+$ and graph $G$, we denote by $\alpha G$,
%the graph obtained by multiplying the capacity of each edge of $G$
%by $\alpha$.
For completeness, we include a proof in the appendix (\ref{app:rerouting}).

\begin{lemma}[\cite{treewidth}]
  \label{lem:match}
  Let $D$ be a demand matrix on a given graph $G$ and let $f: V
  \rightarrow V$ be a mapping. If the cut condition is satisfied for
  $D$, and $D'_f$ is routable in $\gamma G$, then the cut condition is
  satisfied for $D_f$ in $(\gamma+1) \cdot G$.
\end{lemma}

We give a useful corollary of the above lemma.

\begin{corollary}
  \label{cor:reroute}
  Let $G=(V,E)$ satisfy the cut-condition for $H=(V,E_H)$ and let
  $A\subseteq V$ be a node-cover in $H$. Then there exists a demand
  graph $I=(A,F)$ such that $2G$ satisfies the cut-condition for
  $I$. Moreover, if $I$ is (integrally) routable in $2G$ with
  congestion $\alpha$, then $H$ is (integrally) routable in $G$ with
  congestion $(1+2\alpha)$.
\end{corollary}

\begin{proof}
  Assume for simplicity that $G$ and $H$ have unit capacities. Let $A
  \subset V$ be a node-cover in $H$. Shrink $A$ to a node $a$ to
  obtain a new supply graph $G'$ and a new demand graph $H'$.  Since
  $A$ is a node-cover all demand edges in $H'$ are incident to
  $a$, and so $H'$ is a star, and $G',H'$ is a single-source instance. For simplicity assume that there are no parallel edges in $H'$;
  if node $u$ has $d > 1$ parallel edges to $a$, then add $d$ dummy
  terminals connected to $u$ with infinite capacity edges in $G'$ and
  replace each $(u,a)$ edge by an edge from a dummy terminal to
  $a$. Let $S \subseteq V\setminus A$ be the set of nodes that have a
  demand edge to $a$ in $H'$. Note that $G'$ satisfies the
  cut-condition for $H'$. %BRUCE Since $A$ is a vertex-cover in $H$, the
  %graph $H'$ is a star with center $a$ and hence a single-source
  %instance.
  Therefore by the maxflow-mincut theorem (or Menger's
  theorem) $H'$ is routable in $G'$ with congestion $1$ and by our
  assumption that the demands are unit valued and capacities are
  integer valued, the flow corresponds to $|S|$ paths, one from each
  node in $S$ to $a$. Now unshrink $a$ to $A$; thus the flow
  corresponds to paths from $S$ to $A$ in $G$. Define a mapping $f: V
  \rightarrow V$ where $f(u) = u$ if $u \in V\setminus S$ (we only care about $u \in A$), and if $u
  \in S$ then $f(u) = v$ if the path from $u$ to $A$ ends in $v \in
  A$. Let $D$ be the demand matrix corresponding to $H$ and $D_f$ be
  demand matrix induced by the mapping $f$. Let $I=(A,F)$ be the
  demand graph induced by $f$. Note then that $D'_f$ corresponds to
  the single-sink flow problem determined by $H'$. Hence by
  Lemma~\ref{lem:match}, $D_f$, and hence $I$, satisfies the cut
  condition in $2G$. We then apply
  Proposition~\ref{prop:compose-routing} to see that if $I$ is
  (integrally) routable in $2G$ with congestion $\alpha$ (which is the
  same as $I$ being routable in $G$ with congestion $2\alpha$), then
  $H$ is (integrally) routable in $G$ with congestion $(1+2\alpha)$
  since $H'$ is integrally routable in $G$ with congestion $1$.
\end{proof}

\subsection{$k$-Outerplanar and $k$-shell Instances}
\label{sec:kouter}

Let $G=(V,E)$ be an embedded planar graph. We define the outer layer
or the $1$-layer of $G$ to be the nodes of $G$ that are on the outer
face of $G$. The $k$-th layer of $G$ is the set of nodes of $V$ that
are on the outer face of $G$ after the nodes in the first $k-1$ layers
have been removed. A {\em $k$-Outerplanar} graph is a planar graph that can
be embedded with at most $k$ layers. We let $V_i$ denote the nodes on
the $i$-th layer.  We are interested in multiflows in planar graphs
when at least one terminal of each demand pair is lies in one of the
outer $k$ layers; we call such instances {\em $k$-shell instances}.  Let
$H=(V,F)$ be a demand graph.

\begin{theorem}[Okamura-Seymour~\cite{OS}]
\label{thm:os} Let $G$ be a planar graph and $H$ be a demand
graph with all terminals on a single face. If $H$ satisfies the
cut-condition, then there is a half-integral routing of $H$ in $G$.
Moreover if $G+H$ is Eulerian, $H$ is integrally routable in $G$.
\end{theorem}

\begin{theorem}[\cite{ChekuriGNRS}]
\label{thm:kop}
If $G$ is a $k$-Outerplanar graph and $H$ satisfies the cut-condition,
then $H$ is fractionally routable in $G$ with congestion $c^k$ for
some universal constant $c$.
\end{theorem}

We can strengthen the above theorem to prove the following on $k$-shell
instances.

\begin{theorem}
\label{thm:kshell}
Let $G$ be a planar graph and let $V' = \cup_{i=1}^k V_i$ be the set
of nodes in the outer $k$ layers of a planar embedding of $G$.
Suppose $H=(V,F)$ is a demand graph where for each demand edge at
least one of the end points is in $V'$.  If $G,H$ satisfy the cut
condition, then $H$ can be fractionally routed in $G$ with congestion
$c^k$ for some universal constant $c$.
\end{theorem}

The proof of Theorem~\ref{thm:kop} relies on machinery from metric
embeddings. Our proof of Theorem~\ref{thm:kshell} also relies on
metric embeddings, and in particular uses recent results \cite{LeeS09}
and \cite{EnglertGKRTT10}. Since these techniques are some what
orthogonal to the primal methods that we use in this paper, we
describe a proof of Theorem~\ref{thm:kshell} in a separate manuscript
\cite{ChekuriSW-unpub}.  Below we relate the integer and fractional
flow-cut gaps for $k$-shell instances.

\begin{theorem}
  \label{thm:kshell2}
  Let $G$ be a planar graph and let $V' = \cup_{i=1}^k V_i$ be the set
  of nodes in the outer $k$ layers of a planar embedding of $G$.  Suppose
  $H=(V,F)$ is a demand graph where for each demand edge at least one
  of the end points is in $V'$. If $H$ is fractionally routable in
  $G$, then it can be integrally routed in $G$ with congestion $6^k$.
\end{theorem}

Combining the above two theorems, we have:
\begin{corollary}
  If $G,H$ is a $k$-shell instance that satisfies the cut-condition,
  then $H$ is {\em integrally} routable in $G$ with congestion
  $c^k$ for some universal constant $c$.
\end{corollary}

We need the following claim as a base case to prove Theorem~\ref{thm:kshell2}.
\begin{claim}
  \label{claim:os}
  Let $G$ be a planar graph and $H=(V,F)$ be a demand graph such that
  for each demand edge at least one end point is on the outer face
  $V_1$. If $G,H$ satisfy the cut-condition, then there is
  an integral routing of $H$ in $G$ with congestion $5$.
\end{claim}
\begin{proof}
  We observe that $V_1$ is a node cover in $H=(V,F)$. Therefore, by
  Corollary~\ref{cor:reroute} there is a demand graph $I=(V_1,F')$
  such that $2G,I$ satisfy the cut condition. By the Okamura-Seymour
  theorem (Theorem~\ref{thm:os}), $I$ is integrally routable in $2G$
  with congestion $2$. Therefore, by Corollary~\ref{cor:reroute}, $H$
  is integrally routable in $G$ with congestion $5$.
\end{proof}

\begin{proof}[of Theorem~\ref{thm:kshell2}]
  We prove the theorem by induction on $k$.
  The base case of $k=1$ follows from Claim~\ref{claim:os}.

  Assuming the hypothesis for $j<k$, we prove it for $j=k$. Let
  $H=(V,F)$ be a demand graph that is fractionally routable in $G$ and
  such that each demand edge is incident to a node in the outer $k$
  layers. Let $H_k=(V,F_k)$ be the subgraph of $H$ induced by the
  demand edges $F_k \subseteq F$ that are incident to at least one
  node in $V_k$ and moreover the other end point is not in $V_1 \cup
  \ldots \cup V_{k-1}$.  We obtain a new supply graph $G'$ by
  shrinking the nodes in $\cup_{i=1}^{k-1} V_i$ to a single node
  $v$. Note that $H_k$ is fractionally routable in $G'$ as well. Fix
  some arbitrary routing of $H_k$ in $G'$. Partition $F_k$ into
  $F_k^a$ and $F_k^b$ as follows. $F_k^a$ is the set of all demands
  that route at least half their flow through $v$ in $G'$. $F_k^b =
  F_k \setminus F_k^a$.  Thus a demand in $F_k^b$ routes at least half
  its flow in the graph $G'' = G[V \setminus \cup_{i=1}^{k-1}V_i]$. We
  claim that the demand graph $H_k^b=(V \setminus \cup_{i=1}^{k-1}V_i,
  F_k^b)$ is integrally routable in $G''$ with congestion $10$. For
  this we note that $H_k^b$ is fractionally routable in $2G''$ which
  implies that $(2G'',H_k^b)$ satisfies the cut-condition. Moreover
  $V_k$ is the outer-face of $G''$ and each demand in $F_k^b$ has at
  least one end point in $V_k$ and hence we can apply
  Claim~\ref{claim:os}.

  Now consider demands in $F_k^a$ and their flow in $G'$. For
  simplicity assume that the end points of $F_k^a$ are disjoint and
  let $T$ be the set of end points. By doubling the fractional flow in
  $G'$ of each demand $f \in F_k^a$ we get a feasible routing for
  sending one unit of flow from each $t \in T$ to $v$. Thus in $G'$
  there are paths $P_t, t \in T$ where $P_t$ is a path from $t$ to $v$
  and no edge has more than two paths using it. These paths imply that
  the terminals in $T$ can be integrally routed to nodes in
  $\cup_{i=1}^{k-1}V_i$ with congestion $2$ in $G$ (simply unshrink
  $v$). For each demand $f=(u,v) \in F_k^a$ let $f'=(u',v')$ be a
  new demand where $u'$ and $v'$ are the nodes in
  $\cup_{i=1}^{k-1}V_i$ that $u$ and $v$ are routed to. Let these
  demands be $F^c_k$. We claim that $F^c_k$ is fractionally routable in
  $G$ with congestion $3$ - simply concatenate the routing of $F_k^a$
  with the paths that generated $F^c_k$ from $F^a_k$. Now consider
  the demand graph $H'=(V,(F\setminus F_k) \cup F_k^c)$. Since $H=(V,F)$ is
  fractionally routable in $G$ and $F_k^c$ is fractionally routable in $3G$,
  we have that $H'$ is fractionally routable in $4G$. Also, each demand
  in $H'$ has an end point in $\cup_{i=1}^{k-1}V_i$. Therefore, by
  the induction hypothesis, $H'$ is integrally routable in $G$ with
  congestion $4\cdot 6^{k-1}$.

  Routing $H$ as above consists of routing $H'$, the routing of $F_k^a$, and the
  routing of the demands in $F_k^b$ to the outer $k-1$ layers; adding up
  the congestion for each of these routings as shown above, we see
  that $H$ is routable in $G$ with congestion $4\cdot 6^{k-1} + 10 + 2
  \le 6^k$ for $k \ge 2$.  This proves the hypothesis for $k$.
\end{proof}

\subsection{Flow-Cut Gap and Node Cover size of Demand Graph}
\label{sec:gunluk}

Linial, London and Rabinovich \cite{LLR} and Aumann and Rabani
\cite{AR} showed that if the supply graph $G=(V,E)$ satisfies the
cut condition for a demand graph $H=(V,E_H)$, then $H$ is routable
in $G$ with congestion $O(\log k)$ where $k = |E_H|$; to obtain this
refined result (instead of an $O(\log n)$ bound), \cite{LLR,AR} rely
on Bourgain's proof of the distortion required to embed a finite
metric into $\ell_1$. G\"{u}nl\"{u}k \cite{Gunluk07} further refined the
bound and showed that the flow-cut gap is $O(\log k^*)$ where $k^*$
is the size of the smallest node cover in $H$; recall a {\em node
cover} is a subset $S$ of nodes for which every edge of $H$ has at
least one endpoint in $S$. For example if $k^* = 1$, then $H$
induces a single-source problem for which the flow-cut gap is $1$.
G\"{u}nl\"{u}k's argument requires a fair amount of technical reworking of
Bourgain's proof. Here we give a simple and insightful proof via
Lemma~\ref{lem:match}, in particular Corollary~\ref{cor:reroute}.

\begin{theorem}
  Let $G=(V,E)$ satisfy the cut-condition for $H=(V,E_H)$ such that
  $H$ has a node-cover of size $k^*$. Then $H$ is routable in $G$
  with congestion $O(\log k^*)$.
\end{theorem}

\begin{proof}
  Let $A \subset V$ be a node-cover in $H$ such that $|A| = k^*$.
  We now apply Corollary~\ref{cor:reroute} which implies that
  there is a demand graph $I=(A,F)$ such that $2G$ satisfies the
  cut-condition for $I$. Moreover if $I$ is routable in $2G$ with
  congestion $\alpha$ then $H$ is routable in $G$ with congestion $(1+2\alpha)$.
  Note that $I$ is a demand graph with at most $(k^*)^2$ edges, therefore, it
  is routable in $2G$ with congestion $O(\log k^*)$ \cite{LLR,AR}. Hence $H$
  is routable in $G$ with congestion $O(\log k^*)$.
\end{proof}

\subsection{Multiflows with terminals on $k$ faces of a planar graph}
\label{sec:outerface} Lee and Sidiropoulos \cite{LeeS09} recently
gave a powerful methodology via their {\em peeling} lemma to reduce
the flow-cut gap question for a class of instances to other
potentially simpler class of instances. Using this they reduced the
flow-cut gap question for minor-free graphs to planar graphs and
graphs closed under bounded clique sums. One of the applications of
their peeling lemma is the following result.  Let $G$ be an embedded
planar graph and $H$ be a demand graph such that the endpoints of
edges in $H$ lie on at most $k$ faces of $G$. If $G$ satisfies the
cut-condition, then $H$ is (fractionally) routable in $G$ with
congestion $e^{O(k)}$. Their proof extends to graphs of bounded
genus and relies on the non-trivial peeling lemma. Here we give a
simple proof with a stronger guarantee for the planar case, again
using Lemma~\ref{lem:match}.

\begin{theorem}
  Let $G=(V,E)$ be an embedded planar graph and $H=(V,F)$ be a demand
  graph such that the endpoints of edges in $H$ lie on at most $k$
  faces of $G$. If $G$ satisfies the cut-condition, then $H$ is
  routable in $G$ with congestion $3k$. Moreover if $G$ and $H$ have
  integer capacities and demands respectively, then there is an
  integral flow with congestion $5k$.
\end{theorem}
\begin{proof}
  Let $V_1, \ldots, V_k$ be the node sets of the $k$ faces on which
  the demand edges are incident to. Let $F_i \subseteq F$ be the edges
  in $H$ that have at least one end point incident to a node in $V_i$
  and let $H_i=(V,F_i)$ be the demand graph induced by $F_i$.  Note
  that $G$ satisfies the cut-condition for $H_i$. Clearly $V_i$ is a
  node-cover for $H_i=(V,F_i)$ and hence by
  Corollary~\ref{cor:reroute}, there is a demand graph
  $I_i=(V_i,F'_i)$ such that $2G$ satisfies the cut-condition for
  $I_i$. Note, however, that $I_i$ is an Okamura-Seymour instance in
  that all terminals lie on a single face. Hence $I_i$ is routable in
  $2G$ with congestion $1$ and is integrally routable in $2G$ with
  congestion $2$. Hence, by Corollary~\ref{cor:reroute}, $H_i$ is
  routable in $G$ with congestion $3$ and integrally routable in $G$
  with congestion $5$. By considering $H_1, \ldots, H_k$ separately,
  $H$ is routable in $G$ with congestion $3k$ and integrally with
  congestion $5k$.
\end{proof}

\noindent {\bf Remark:} We observe that a bound of $k$ is easy via
the Okamura-Seymour theorem if for each demand edge the two end
points are incident to the same face. The rerouting lemma allows us
to easily handle the case when the end points may be on different
faces; in fact the proof extends to the case when nodes on $k$ faces
in $G$ form a node cover for the demand graph $H$. The above proof
can be easily extended to graphs embedded on a surface of genus $g$
to show a flow-cut gap of $3\alpha_g k$ where $\alpha_g$ is the gap
for instances in which all terminals are on a single face.

\section{Lower-bound on flow-cut gap in series-parallel graphs}

\label{app:lower}
%{Fractional instances that require congestion up to $2$}
Lee and Raghavendra have shown in~\cite{LeeR07} that the flow-cut gap
in series-parallel graph can be arbitrarily close to $2$; this lower
bound matches an upper bound found previously
\cite{ChakrabartiJLV08}. They introduce a family of supply graphs and
prove the required congestion using the theory of metric embeddings.
In particular, their lower bound is shown by a construction of a
series parallel graph and a lower bound on the distortion required to
embed the shortest path metric on the nodes of the graph into
$\ell_1$.

In this section we give a different proof of the lower bound. We use
the same class of graphs and the recursive construction of Lee and
Raghavendra \cite{LeeR07}. However, we use a primal and direct
approach to proving the lower bound by constructing demand graphs that
satisfy the cut-condition but cannot be routed. We lower bound the
congestion required for the instances by exhibiting a feasible
solution to the dual of the linear program for the maximum concurrent
multicommodity flow in a given instance. This is captured by the
standard lemma below which follows easily from LP duality.

\begin{lemma}
  \label{lem:dual-lb}
  Let $G=(V,E)$ and $H=(V,F)$ be undirected supply and demand graphs
  respectively with $c_e$ denoting the capacity of $e \in E$ and
  $d_f$ denoting the demand of $f \in F$. Then the minimum congestion
  $r$ required to route $H$ in $G$ is given by
  \[
  \min_{\ell:E\rightarrow \cR^+} \frac{\sum_{f \in F} d_f
    \ell_f}{\sum_{e \in E} c_e \ell_e}
    \]
    where $\ell_f$ is the shortest path distance between the end
    points of $f$ in $G$ according to the non-negative length function
    $\ell:E \rightarrow \cR^+$.
\end{lemma}

\begin{corollary}
  \label{cor:std-lb}
  By setting $\ell_e = 1$ for each $e \in E$, the minimum congestion
  required to route $H$ in $G$ is at least $D/C$ where $C = \sum_{e
    \in E} c_e$, and $D = \sum_{f \in F} d_f \ell_f$, where $\ell_f$ is
  the shortest path distance in $G$ between the end points of $f$
  (with unit-lengths on the edges).
\end{corollary}

We refer to $C$ and $D$ as defined in the corollary above as the
total capacity and the total demand length respectively. We refer
to the lower bound $D/C$ as the {\em standard} lower bound.

\paragraph{Construction of the Lower Bound Instances:}
The family of graphs presented in \cite{LeeR07} is obtained
from a single edge by a simple operation, which consists in replacing
every edge by a $K_{2,m}$ graph, as shown in Figure~\ref{fig:2gap-supply}.

\begin{figure}[h]
  \centering
  \includegraphics[width=5in]{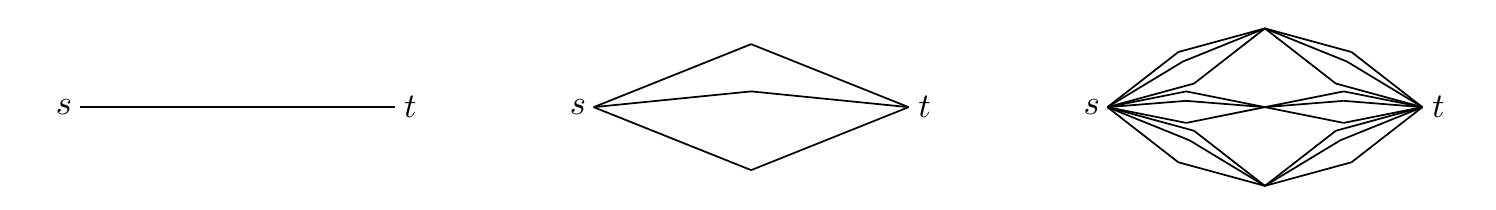}
  \caption{Supply graph for gap example. Replace each edge by $K_{2,m}$, here
$m=3$.}
  \label{fig:2gap-supply}
\end{figure}

\begin{figure}[h]
  \centering
  \includegraphics[height=1.5in]{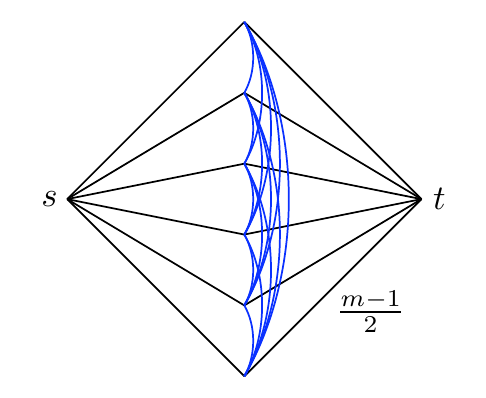}
  \caption{Building block for gap example.}
  \label{fig:2gap-demand}
\end{figure}

Assuming $m$ is even, consider the following demands and capacities
for a $K_{2,m}$ graph: a complete demand graph on the $m$ nodes of
degree $2$, with demand $1$ on each edge, and a capacity $(m-1)/2$ on
each supply edge. See Figure~\ref{fig:2gap-demand}. One easily checks
that the cut condition holds for this instance. Moreover, the total
capacity is $2m \cdot (m-1)/2 = m(m-1)$, and the total demand length
is $2 \cdot m(m-1)/2 = m(m-1)$; hence the standard lower bound $D/C =
1$, and in fact the instance is routable. The simple yet important
observation is that for any central cut separating $s$ from $t$, the
surplus of the cut is equal to $m(m-1)/2-k(m-k)$, when the cut
separates the degree $2$ nodes into sets of cardinality $k$ and
$m-k$. The minimum surplus is attained when $k=m/2$, where it is
$m(m-2)/4$. We call this instance $I(s,t)$. Based on this we define
an instance $I_c(s,t)$ for any real $c > 0$ as the instance obtained
from $I(s,t)$ by multiplying all demands and capacities of $I(s,t)$
by $\frac{c}{m(m-2)/4}$. The effect of this is that $I_c(s,t)$ satisfies
the cut condition, has a standard lower bound of $1$ and is routable, and
the surplus of any central cut separating $s$ and $t$ is at least $c$.

We summarize the properties of $I_c(s,t)$ in the following lemma, with
the last property being a crucial one.

\begin{lemma}
The instance $I_c(s,t)$ has the following properties:
\begin{enumerate}
\item It has $2m$ supply edges, each of capacity
  $\frac{c}{m(m-2)/4}\frac{m-1}{2}$, and it has $\frac{m(m-1)}{2}$
  demand edges, each with demand $\frac{c}{m(m-2)/4}$.
\item The standard lower bound for $I_c(s,t)$ is $1$ and the instance is
  routable.
\item The minimum surplus of any cut separating $s$ from $t$ is equal
  to $c$.
\item Adding an $st$ demand edge of demand $c$ to $I_c(s,t)$
  preserves the cut condition but the standard lower bound for this modified
  instance is $(1+\frac{m-2}{2(m-1)})$.
\end{enumerate}
\end{lemma}
\begin{proof}
  The first three facts follow directly from the description of
  $I_c(s,t)$.  The surplus of any $st$ cut in $I_c(s,t)$ is $c$, and
  $I_c(s,t)$ satisfies the cut condition. Hence, it follows that
  adding an $st$ demand edges with demand $c$ preserves the cut
  condition. We note that for $I_c(s,t)$, the total capacity $C = 2m
  \cdot \frac{c}{m(m-2)/4}\frac{m-1}{2} = \frac{4c(m-1)}{(m-2)}$, this
  is also equal to the total demand length $D$ for $I_c(s,t)$. Adding
  an $st$ demand edge with demand value $c$ increases the total demand
  length by $2c$, and therefore the standard lower bound for the
  modified instance is $\frac{C + 2c}{C} = (1+\frac{m-2}{2(m-1)})$.
\end{proof}

As in \cite{LeeR07}, the goal is to show that the congestion lower
bound obtained in the construction above can be amplified by
iteration.  As a first step, the next lemma shows that replacing an
edge of the supply graph by an instance of type $I_c(s,t)$ maintains
the cut condition. Given a multiflow instance $G,H$ and a supply edge
$uv$ in $G$ we can obtain a new instance $G',H'$ by replacing $uv$ by
the instance $I_c(s,t)$ where we identify $u$ with $s$ and $t$ with
$v$. Note that all $H'$ contains all the demand edges of $H$
and the demand edges of $I_c(s,t)$.

\begin{lemma}\label{lem:cutcond}
  Let $I=G,H$ be a multiflow instance that verifies the cut condition. Let
  $uv$ be a supply edge in $G$ with capacity $c > 0$. Let $I'=G',H'$ be
  a new instance obtained by replacing $uv$ by $I_c(s,t)$.
   Then $G',H'$ verifies the cut condition.
\end{lemma}
\begin{proof}
\iffalse
  Let $A = \{s,t,v_1,v_2, \ldots, v_m\}$ be the nodes of the supply
  graph of $I_c(s,t)$. Consider any {\em central} cut $S$ in $G'$.  If
  $S \cap A = \emptyset$ or if $S \cap A = A$, then the surplus of the
  cut $S$ in $G'$ is the same as its surplus in $G$. Thus we consider
  the case when $S \cap A$ and $(V\setminus S) \cap A$ are non-empty.
  For $S$ to be central, there are two cases.  First, $S$ is a single
  node from $\{v_1,\ldots, v_m\}$, in which case its surplus in $I'$ is
  non-negative since it is non-negative in $I_c(s,t)$. Second,
  without loss of generality $S$ contains $s$ and not $t$, and
  some nodes $A' \subseteq \{v_1,\ldots, v_m\}$.
\fi
The surplus of any cut separating $u$ from $v$ in $I'$ is at
least as big as the surplus of the corresponding cut in $I$, since the
surplus of any cut separating $s$ from $t$ in $I_c(s,t)$ is at least
$c$. Also, the surplus of any cut in $I'$ that does not separate
$u$ from $v$ is equal to the surplus of the corresponding cut in $I$,
plus the nonnegative surplus of some cut in $I_c(u,v)$.
\end{proof}

From any instance, we build an instance requiring a larger congestion
using the following transformation: We replace \emph{each} supply edge
by an instance of type $I_c(u,v)$. The next theorem gives a lower
bound on the required congestion in the transformed instance.

\begin{theorem}
  \label{thm:amplify}
  Let $I=G,H$ be an instance verifying the cut condition, with total
  capacity $C$ and total demand length $D$. Let $I'$ be the instance obtained
  by replacing each supply edge $uv$ of capacity $c_{uv}$ by an
  instance $I_{c_{uv}}(u,v)$. Then the transformed instance $I'=G',H'$
  verifies the cut condition, and the standard lower bound for $I'$
  is $1+\frac{D}{C} \cdot \frac{m-2}{2(m-1)}$.
\end{theorem}
\begin{proof}
  The fact that $I'$ verifies the cut condition is a direct
  consequence of Lemma~\ref{lem:cutcond}.

  Let us first compute the total capacity of $G'$. Since each supply
  edge of capacity $c$ in $I$ is replaced with $2m$ supply edges of
  capacity $\frac{c}{m(m-2)/4}\frac{m-1}{2}$, the total capacity of
  $G'$ is $C'=C\frac{2m}{m(m-2)/4}\frac{m-1}{2}=C\frac{4(m-1)}{m-2}$.

  The demand edges in $H'$ either exist in $H$, or are added by the
  transformation, i.e. they are internal to some $I_c(u,v)$ instance
  that replaces a supply edge $uv$. The total demand length in $I'$ can
  therefore be decomposed into the part corresponding to the demand
  edges that exist in $I$, and the total demands internal to each
  $I_c(u,v)$ instance. In each $I_c(u,v)$ instance, the total demand length
  is equal to the total capacity. So the sum of demand lengths internal
  to each $I_c(u,v)$ instance is equal to the total capacity,
  which is the total capacity of $G'$. On the other hand,
  the demand edges that exist in $I$ have the same demand, but the
  shortest path of each such edge has exactly doubled in the transformed
  instance. It follows that $D' = C'+2D$.

  Therefore, the standard lower bound for $I'$ is is
  $\frac{D'}{C'}=\frac{C'+2D}{C'}=1+\frac{2D}{C}\frac{m-2}{4(m-1)} =
  1+\frac{D}{C}\frac{m-2}{2(m-1)}$.
\end{proof}

\begin{figure}[ht]
  \centering
  \includegraphics[width=6in]{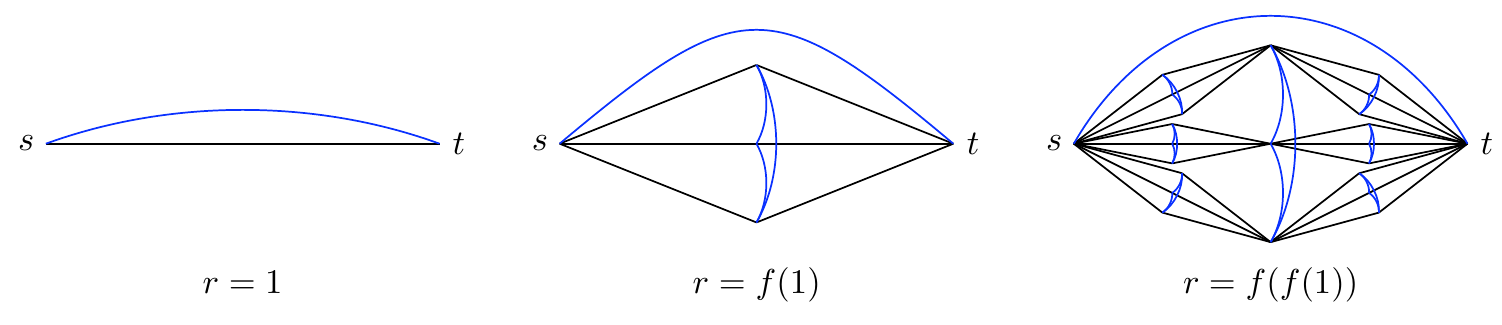}
  \caption{The flow-cur gap $r$ grows bigger with each iteration.}
  \label{fig:2gap-proof}
\end{figure}

Thus, Theorem\ref{thm:amplify} can be used to amplify the flow-cut
gap.  In particular, if the standard lower bound yields a gap of $x$
for an instance $I$ then, one obtains an instance with standard lower
bound yielding a gap of $f(x)=1+x\frac{m-2}{2(m-1)}$. We can iterate
this process $k$ times yielding instances with flow-cut gap
$f^{(k)}(x)$ where $f^{(k)}$ is the function $f$ composed $k$
times. We note that
$$
f^{(k)}(x) = 1+\frac{(m-2)}{2(m-1)}+\cdots+\left(\frac{(m-2)}{2(m-1)}\right)^{k-1}+x\left(\frac{(m-2)}{2(m-1)}\right)^k,
$$
and that this converges to
$$
\lim_{k\rightarrow\infty}f^{(k)}(x) = 2-\frac{2}{m},
$$
for any $x$.

\begin{theorem}
  For any $\varepsilon > 0$, there is a series-parallel graph instance for which
  the flow-cut gap is $2-\varepsilon$.
\end{theorem}
\begin{proof}
  We apply the iterated construction starting with the instance
  consisting of a single edge of capacity $1$ as the supply graph, and
  a demand graph consisting of the same edge with demand $1$; clearly,
  the standard lower bound for this instance is $1$. See
  Figure~\ref{fig:2gap-proof}.  We observe that the construction
  preserves the property that the supply graphs are series
  parallel. Thus, if we set $m = \lceil 4/\varepsilon \rceil$ in
  $I(s,t)$, the iterated construction yields a series parallel graph
  instance with flow-cut gap arbitrarily close to $2-\varepsilon/2$, and
  hence one can choose a sufficiently large $k$ such that iterating
  the construction $k$ times gives an instance with gap at least
  $2-\varepsilon$.
\end{proof}

\iffalse
%BRUCE:  Remember that in other files we have:

1. 2g lemma (not yet fixed) which looks at K2m instances, and routes twice the st demand, in 2G.

2.  congestion 2 result for routing in disjoint path instances. this uses the red edge pushing business.

\fi

\paragraph{Acknowledgements:} This work was partly done while the
first two authors were at Lucent Bell Labs; they acknowledge support
from an ONR grant N00014-05-1-0256. C.~Chekuri is also partly
supported by NSF grant CCF-0728782. We thank Anupam Gupta and Sanjeev
Khanna for several discussions and ideas along the course of this
work. We thank an anonymous reviewer for pointing out a minor error in
an earlier proof of Theorem~\ref{thm:kshell2}.

\appendix
\section{Appendix}

\subsection{Proof of the Rerouting Lemma \cite{treewidth}}
\label{app:rerouting}

For completeness we give the proof from \cite{treewidth}.

\begin{proof}
  Consider a cut $\delta_G(S)$ in $G$ for some $S \subset V$.
  Let $D(S)$ denote the total demand across this cut.
  %Recall that $c(S)$ is the total capacity of the edges in $\delta(S)$.
%  For a demand matrix $M$, let $M(S) = \sum_{uv: u \in S, v \not \in S}
%  M(uv)$ denote the demand of $M$ that crosses $S$.
  Since $D$ satisfies the cut condition $|\delta_G(S)| \ge D(S)$ for all $S
  \subset V$.  Also, since $D'_f$ is routable in $\gamma G$, $\gamma
  |\delta_G(S)| \ge D'_f(S)$ for all $S \subset V$.

  To prove the lemma we need to show that $(\gamma+1) |\delta_G(S)| \ge
  D_f(S)$.  From the above inequalities, it suffices to show
  that $D_f(S) \le D(S) + D'_f(S)$.  Let $X_S$ denote the set of all
  unordered pairs of nodes $uv$ such that $u$ and $v$ are separated by
  $S$, that is $|\{u,v\} \cap S| = 1$. We can write $D_f(S)$ as
  $\sum_{uv: f(u)f(v) \in X_S} D(uv)$. For each pair $uv$ such that
  $f(u)f(v) \in X_S$, we charge $D(uv)$ to either $D(S)$ or $D'_f(S)$
  such that there is no overcharge. This will complete the argument.

  We consider two cases. If $uv \in X_S$ then we charge $D(uv)$ to
  $D(S)$.  Note that $\sum_{uv \in X_S} D(uv) = D(S)$ and hence we do
  not over charge $D(S)$. If $uv \not \in X_S$, then either $uf(u) \in X_S$
  or $vf(v) \in X_S$ but not both. In $uf(u) \in X_S$ we charge $D(uv)$ to
  $u$, otherwise to $v$. We observe that the total charge to
  a node $u$ is at most $D'_f(uf(u))$ and it is charged only if
  $uf(u) \in X_S$. Hence the total charge to $D'_f(S)$ is not exceeded
  either.
\iffalse PAINFUL ALGEBRAIC PROOF.
  Then we have that
  \begin{eqnarray*}
    D_f(S) & = & \sum_{uv: f(u)f(v) \in X_S} D(uv) \\
    & = & \sum_{uv \in X_S: f(u)f(v) \in X_S} D(uv) + \sum_{uv \not \in X_S: f(u)f(v) \in X_S} D(uv) \\
    & \le & \sum_{uv \in X_S} D(uv) + \sum_{uv \not \in X_S: f(u)f(v) \in X_S} D(uv) \\
    & \le & c(S) +  \sum_{u \in S, f(u) \in V\setminus S} D'_f(uf(u)) + \sum_{v \in V\setminus S, f(v) \in S} D'_f(vf(v)) \\
    & \le & c(S) + D'_f(S) \\
    & \le & c(S) + \gamma c(S) \\
    & \le & (\gamma+1) c(S).
  \end{eqnarray*}
\fi
\end{proof}

\end{document}

\subsection{Proof of Theorem~\ref{thm:kfaces}}
\label{app:kfaces}

\begin{proof}[Sketch]
  Let $F_1, \ldots, F_k$ be the node sets of the $k$ faces on which
  the terminals lie.  Consider $F_1$ and let $H_1=(V, E_1)$ be the
  demand graph induced by the edges in $H$ such that at least one endpoint
  of the edge is in $F_1$. Let $S_1$ be the terminals in
  $H_1$ not in $F_1$; these are terminals with a demand edge to some
  terminal in $F_1$. For $t \in S_1$ let $d(t)$ be the amount of demand incident
  to $t$ in $H_1$. Since $G$ satisfies the cut-condition for $H_1$, by
  the maxflow-mincut theorem, we can route the demand from terminals
  in $S_1$ to $F_1$ --- note that the demand $d(t)$ from $t \in S_1$
  may be split and reach different nodes in $F_1$, and not the
  desired destination. Once the demand from $S_1$ reaches $F_1$ we
  wish to reroute this flow to the desired destination. The key point
  is that this rerouting involves nodes only on $F_1$  and so this defines
  an OS instance on the face corresponding to $F_1$. Moreover, by
  Lemma~\ref{lem:match}, $2G$ satisfies the cut condition for this
  rerouting problem. By the Okamura-Seymour theorem, $2G$ admits a
  feasible flow for the rerouting problem. Composing the two flows
  shows that $H_1$ is routable in $3G$. Repeating this with each of
  the other faces, we see that $H$ is routable in $G$ with congestion
  $3k$. For the integral flow, observe that the first phase routing of
  $S_1$ to $F_1$ is routable integrally. The second phase via the
  Okamura-Seymour theorem gives a half-integral flow in $2G$ which
  implies an integral flow in $4G$. Hence $H_1$ is integrally routable
  in $5G$ and thus $H$ is integrally routable in $G$ in $5kG$.
\end{proof}